\documentclass[a4paper,12pt,fleqn]{article}
\usepackage[margin=1.25in]{geometry}
\usepackage[utf8]{inputenc}

\usepackage[USenglish]{babel}
\usepackage{enumitem}

\usepackage{setspace}

\usepackage{amssymb}
\usepackage{amsfonts}
\usepackage{amsmath}
\usepackage{mathabx}
\usepackage{bbm}

\usepackage{comment}
\usepackage{graphicx}
\usepackage{titlesec}
\usepackage{amsthm}
\usepackage{mathrsfs}  

\usepackage{relsize}
\usepackage{exscale}

\usepackage{float}

\newcommand{\bsf}[1]{\textnormal{\sffamily\bfseries #1}}

\usepackage{cellspace}

\newcommand{\powerset}{\raisebox{.15\baselineskip}{\large\ensuremath{\wp}}}
\usepackage{xcolor}
\definecolor{darkred}{rgb}{0.55, 0.0, 0.0}
\definecolor{darkscarlet}{rgb}{0.34, 0.01, 0.1}
\definecolor{Gold1}{rgb}{0.9, 0.7, 0.1}
\definecolor{green4}{rgb}{0.0, 0.26, 0.15}
\definecolor{green3}{rgb}{0.0, 0.5, 0.0}

\definecolor{maastricht}{rgb}{0.0, 0.18, 0.39}
\definecolor{mycolor}{rgb}{0.0, 0.18, 0.39}
\definecolor{bg}{rgb}{0.97, 0.97, 0.94}

\usepackage[most]{tcolorbox}
\newtcolorbox{mybox}[2][]{
colback=bg,
colframe=black!75!black,
fonttitle=\bfseries,
coltitle=white,
colbacktitle=bg,
enhanced,
attach boxed title to top left={yshift=-1.2mm, xshift=2mm},
title=#2,
#1}

\usepackage{tikz}

\usetikzlibrary{fit,calc,shapes.callouts,shapes.arrows,matrix,positioning,arrows}

\usetikzlibrary{calc,shapes.callouts,shapes.arrows}

\usepackage{pgfplots}
\pgfplotsset{compat=1.18}

\usepackage[authoryear]{natbib}

\usepackage[
  colorlinks=true,
  citecolor=green3,
  linkcolor=maastricht,
  urlcolor=darkscarlet,
  runcolor=darkscarlet,
  pagebackref=true
]{hyperref}
\usepackage[nameinlink]{cleveref}

\crefdefaultlabelformat{#2\textsc{#1}#3}
\Crefname{definition}{Definition}{\textsc{Definitions}}
\Crefname{theorem}{Theorem}{\textsc{Theorems}}
\Crefname{section}{Section}{\textsc{Sections}}
\Crefname{lemma}{Lemma}{\textsc{Lemmas}}
\Crefname{corollary}{\textsc{Corollary}}{\textsc{Corollaries}}
\Crefname{figure}{Figure}{\textsc{Figures}}
\Crefname{example}{Example}{\textsc{Examples}}
\Crefname{remark}{Remark}{\textsc{Remark}}

\usepackage{thmtools}

\renewcommand\thmcontinues[1]{Continued}

\newcommand*{\fullref}[1]{\hyperref[{#1}]{\autoref*{#1} \nameref*{#1}}}

\usepackage{titletoc}
\newtheoremstyle{boldnote}%
  {3pt}{3pt}
  {\normalfont}%
  {}%
  {\bfseries}
  {.}%
  { }%
  {\thmname{#1}~\thmnumber{#2}\textbf{\thmnote{ (#3)}}}

  \newtheoremstyle{boldnote-italic}%
  {3pt}{3pt}
  {\itshape}%
  {}%
  {\bfseries}
  {.}%
  { }%
  {\thmname{#1}~\thmnumber{#2}\textbf{\thmnote{ (#3)}}}

\theoremstyle{boldnote-italic}
\newtheorem{theorem}{Theorem}

\newtheorem{corollary}{{Corollary}}

\newtheorem{lemma}{{Lemma}}

\newtheorem{proposition}{Proposition}

\theoremstyle{boldnote}
\newtheorem{definition}{{Definition}}[section]

\newtheorem{example}{{Example}}
\newtheorem{remark}{{Remark}}[section]

\usepackage{scalerel}

\usepackage{dashbox}

\usetikzlibrary{calc}
\usetikzlibrary{matrix}
\usetikzlibrary{positioning}
\usetikzlibrary{arrows}
\usepackage{epigraph}

\usepackage{tcolorbox}

\usepackage[labelfont=bf]{caption}

\usetikzlibrary{calc}
\usetikzlibrary{matrix}
\usetikzlibrary{positioning}
        
    \tikzset{
  every overlay node/.style={
    anchor=north west, inner sep=0pt,
  },
}
\begin{document}
\spacing{1.5}

\pagenumbering{roman}
\thispagestyle{empty}

\title{\vspace{-3cm}{\sc \Large Rationalizable Behavior in Matching with Externalities}\thanks{We greatly appreciate the comments of and discussions with Pierpaolo Battigalli,  Laura Doval, Bettina Klaus, Pierfrancesco Guarino, Nathan Yoder, Marek Pycia, Christian Seel, Ziwei Wang and the participants of   the 14th Conference on Economic Design (Essex, UK), the $29^{th}$ Coalition Theory Network Workshop (CTN, Padova, Italy), the 2025 SAET Conference (Ischia, Italy), Match-Up 2026 (Paris, France), SAET-EWET 2026 (Venice, Italy), SING 2026 (Naples, Italy), SSCW 2026 (Tokyo, Japan) and The 2026 Centre for Institution and Mechanism Design Symposium (Macau, China). The usual disclaimers apply. Antonio  Nicolò acknowledges financial support  by EU  Next generation under grant PRIN call PNRR 2022, project \lq\lq Incentivizing participation of compatible pairs in Kidney Paired Exchange Programs\rq\rq n. P2022P5CHH.
Riccardo Saulle acknowledges financial support  by Unicredit Foundation.
}}
 \author{
 Antonio Nicolò\thanks{University of Manchester and University of Padova. E-mail: 
\texttt{antonio.nicolo@manchester.ac.uk}} \and
Pietro Salmaso\thanks{University of Naples Federico  II. E-mail: \texttt{pietro.salmaso@unina.it}} \and
 Riccardo D.  Saulle\thanks{University of Padova. E-mail: \texttt{riccardo.saulle@unipd.it}}
}

\thispagestyle{empty}
\maketitle

\vspace{-1cm}

\begin{abstract}
In many matching markets, agents care not only about their own partners but also about the matches formed by others. With externalities, stability depends on what agents believe would happen after a deviation. We introduce rationalizable conjectures: beliefs that survive iterated elimination, in the spirit of rationalizability in non-cooperative games. These beliefs define conjecture-rationalizable stability, a solution concept that always exists and coincides with Gale--Shapley stability when externalities are absent. We also introduce rationalizable matchings as a non-equilibrium solution concept and show that every conjecture-rationalizable stable matching is rationalizable.
Finally, we provide an epistemic foundation: rationalizability is behaviorally implied by pairwise rationality and common belief in pairwise rationality, while conjecture-rationalizable stability additionally requires belief correctness.

\bigskip

\noindent 
\textbf{Keywords:} Matching, Externalities, Stability,  Rationalizability.\\
\noindent \textbf{JEL Codes:}  C71, C72, C78, D62
\end{abstract}

\newpage
\thispagestyle{empty} \newpage

\titleformat{\section}[hang]
{\normalfont\large\fillast\sc\color{black}}
{\scshape  \oldstylenums{\thesection}}
{1ex minus .1ex}{\large}
\titlespacing{\section}{3pc}{*3}{*2}[3pc]

\titleformat{\subsection}[hang]
{\normalfont\fillast\bf\color{black}}
{\scshape  \oldstylenums{\thesubsection}}
{1ex minus .1ex}{\normalsize}
\titlespacing{\subsection}{3pc}{*3}{*2}[3pc]

\titleformat{\subsubsection}[hang]
{\normalfont\large\fillast\sc\color{black}}
{\scshape  \oldstylenums{\thesubsubsection}}
{1ex minus .1ex}{\normalsize}
\titlespacing{\subsubsection}{3pc}{*3}{*2}[3pc]

{\hypersetup{linkcolor=black}\tableofcontents}

\thispagestyle{empty}
\clearpage
\pagenumbering{arabic} 
\setcounter{page}{1}

\vspace{-0.5em}
\begin{quote}
 {\it\footnotesize \emph{Proteus:} If hearty sorrow
Be a sufficient ransom for offense,
I tender ’t here. I do as truly suffer
As e’er I did commit.

\emph{Valentine:}  Then I am paid,
And once again I do receive thee honest.
[...]
And that my love may appear plain and free,
All that was mine in Silvia I give thee.
}

 \hfill ---{\footnotesize The Two Gentlemen of Verona, Act 5, Scene 4 (W. Shakespeare)}---

\end{quote}
\vspace{-0.5em}

\section{Introduction}

\subsection*{Motivation \& Results}

\noindent {\sc The Two Gentlemen of Verona}, often regarded as Shakespeare's earliest comedy, stages a marriage problem with externalities. Proteus and Valentine are two friends who both desire Silvia. In the seminal marriage model of \citet{Gale1962}, this would be a familiar problem of competing preferences over partners, with agents on each side ranking potential partners on the other. But Shakespeare's final scene suggests something more: each man cares not only about whether he marries Silvia, but also about what happens to the other. The object of choice is not merely a spouse, but a matching. The example is stylized, but the phenomenon is pervasive. Potential employees may evaluate a job offer differently depending not only on the job itself but also on the employment prospects of their spouses. A hospital's value from hiring a cardiologist may fall if a nearby hospital has just hired a more prominent cardiologist, while its value from hiring a neurologist or an oncologist may rise because differentiation avoids direct competition.

\medskip

\noindent This is exactly the feature ruled out by the classical marriage model of \citet{Gale1962}. Gale and Shapley’s central contribution was to reduce matching to a stability problem over individual preferences.\footnote{The reduced-form approach abstracts from the particular process through which agents meet, bargain, or learn. Although we acknowledge the relevance of non-cooperative investigation, we emphasize the advantages of the reduced-form's parsimony in capturing complex strategic situations. For a discussion of reduced-form and non-cooperative models, see \citet{LiuQFor}.} An allocation is stable if no matched agent would prefer to remain single, and no pair of agents would both prefer to match with each other rather than stay with their current partners. Because agents evaluate only their own partners, stability is local: to assess a blocking pair, no further counterfactual is needed. The pair carries within itself all the information required to evaluate the deviation.

\medskip

\noindent Externalities break this locality. When preferences are defined over full matchings, a pair contemplating a deviation must evaluate not only the partnership it would form, but also the residual matching that is expected to arise among the remaining agents. A pairwise deviation is therefore only half of a counterfactual: the other half is the residual matching induced outside the deviating pair. Anticipating the behavior of others becomes central to the  definition of stability.

\medskip

\noindent This is the conceptual problem underlying P-stability, the standard way in which  pairwise stability is formulated in matching markets with externalities. Since preferences are defined over full matchings, the statement that a pair blocks is incomplete unless the theory specifies the residual  matching used to evaluate the deviation. P-stability resolves this incompleteness by imposing a simple counterfactual: when a pair deviates, the agents left behind become single, while all other matches remain unchanged. Thus, P-stability makes blocking well defined by selecting a particular residual matching.

\medskip

\noindent This selection is natural and tractable, but it is also a substantive assumption.
 We therefore ask a prior question: which residual matchings are credible? Our answer is not another post-deviation rule, but a rationality criterion for the counterfactuals that make blocking meaningful.

\medskip

\noindent We formalize this idea through conjectural stability.
 Agents  are  equipped  with a set of conjectures about the behavior of the other agents, formally represented as residual  matchings. Then,  a  matching is  said to be conjecturally stable if no agent has an incentive to deviate given her conjectures\footnote{More precisely, an agent deviates from a matching only if she is certain to be better off in every matching she considers possible.} (\Cref{stability}).

\medskip

\noindent Since different conjectures generate different predictions (\Cref{nestedness}), the central issue is which conjectures are justified when preferences are transparent.\footnote{That is, common knowledge in the usual sense.}

\medskip

\noindent To address this issue, we introduce conjecture-rationalizable stability.\footnote{The terminology emphasizes that rationalizability is imposed on the conjectures that complete pairwise deviations. This distinguishes our concept from rationalizable conjectural equilibrium \citep{Rub1994} and from the notions of rationalizable conjectural stability in \citet{Mauleon2021,Mauleon2025}, which build on Rubinstein and Wolinsky's framework and on imperfect information with feedback, rather than on the conjectural-stability tradition initiated by \citet{Sasaki1986,Sasaki1996}.} Its construction mirrors the logic of rationalizability in non-cooperative game theory \citep{Pearce1984,Bernheim1984}. We introduce a conjecture-based notion of dominance (\Cref{conjecturedominance}) and an iterative procedure that eliminates dominated residual matchings from agents' conjectures (\Cref{procedure}). A matching is conjecture-rationalizable stable if no agent or pair can profitably block it under rationalizable conjectures. Blocking must therefore be justified by beliefs that survive iterated reasoning.

\medskip

\noindent A key feature of the construction is partner-dependence. In non-cooperative games, each agent holds beliefs about a fixed set of opponents and their strategies.  In matching environments, by contrast, conjectures are tied to the contemplated partner: changing the partner changes what can be conjectured. If agent $i$ considers deviating with $j$, her conjecture concerns the behavior of the agents outside $\{i,j\}$. If she considers deviating with $k$, it concerns a different residual set. Hence, a single belief about “what others do” is not enough.
\medskip

\noindent Our main result establishes that conjecture-rationalizable stable matchings always exist (\Cref{th1}). Beyond existence, we derive three results that clarify the structure and epistemic underpinnings of the framework. First, we provide a fixed-point characterization of rationalizable conjecture systems (\Cref{th2}), analogous to classical fixed-point characterizations of best reply sets in non-cooperative game theory \citep{Pearce1984}. Second, we define rationalizable matchings (\Cref{rat}), the counterpart of rationalizable strategy profiles in our environment, and show that every conjecture-rationalizable stable matching is rationalizable (\Cref{thRef}). Third, we provide an epistemic foundation for our solution concepts and for P-stability. The epistemic analysis identifies the rationality assumptions underlying our concepts and links them to common belief in rationality and belief correctness (\Cref{epistemic1}, \Cref{epistemic2}). 
It shows that conjecture-rationalizable stability imposes rationalizability on
off-path conjectures and correctness on the conjectures associated with the
realized matching, whereas rationalizable matchings require only the former. Traditional P-stability imposes an ad hoc epistemic requirement, since beliefs must agree with the particular post-deviation matching selected by the P-stability counterfactual (\Cref{epistemic3}).
\medskip

\noindent
Together, these results establish   that stability with externalities is a question not just of profitable deviations but of credible conjectures.
Without externalities, a blocking pair carries its own counterfactual. With externalities, it does not: the pair must also conjecture how the other agents match. Our contribution is to discipline these conjectures by rationalizability, separating those that are merely assumed from those that can be justified by iterated reasoning. Thus, the paper preserves the reduced-form parsimony of matching theory while extending its scope: pairwise deviations are still evaluated without specifying a full extensive form, but the conjectures that complete them are disciplined by non-cooperative reasoning.

\subsection*{Related Work}
\label{lit}
The literature on matching with externalities largely builds on the seminal contribution of \citet{Sasaki1986} who  introduces the  idea of conjectural stability for matching problems and     identifies  three  solution concepts: 
optimistic-stability,  conservative-stability, and P-stability.\footnote{\citet{Sasaki1986} call pessimistic-stability the notion of 
conservative-stability. We follow the terminology of \citet{Bando2016} to avoid confusion with P-stability. Optimistic stability  allows deviation only if agents improve in at least one matching in which they believe, while conservative-stability allows deviation only if agents improve in all matchings in which they believe. \citet{Sasaki1986} show that the three solutions are nested---optimistically-stable matchings are P-stable which in turn are conservatively-stable.} The latter concept is the predominant one in the literature \citep[e.g.,][]{Roth1990}. A matching is said to be P-stable if there are no singles or pairs that block, given that  former partners remain single and all the other agents remain matched as they were prior to the blocking. This notion is  based on a tractable but very  
 specific assumption about how the agents react to a blocking.\footnote{
This assumption recasts the 
outsider-independent dominance relation \citep{Koczy} for models of coalition formation, and it  appears under different names in the literature. Relevant examples are the  $\gamma$-model \citep{Hart},  non-interference \citep{Konishi} and disintegration rule \citep{Bloch}  in coalitional games. See also \cite{Saulle} who impose such assumption on a general model of coalition formation that encompasses matching and coalitional games as particular cases.}\footnote{In \Cref{foundation} we provide an epistemic foundation of P-stability and show that it builds on strong conditions.}
Moreover, it is well established that a P-stable matching does not always exist\footnote{In \Cref{couple}, we discuss an application of matching with couples in which the set of P-stable matchings is empty, and we show that our solution provides a reasonable prediction.}
\citep[e.g.][]{Sasaki1986,Roth1990}. 
This drawback    motivates  \citet{Sasaki1996} to focus on conservative-stability. In line with \citet{Sasaki1996}  our approach utilizes conservative agent behavior. However, unlike \citet{Sasaki1996} where agents’ conjectures are exogenous and treated as primitives of the model we assume that   agents' conjectures depend on  their preferences, and thus they are   endogenous to the model.\footnote{
 \citet{Sasaki1996} discuss the possibility of endogenizing agents' conjectures  by adopting  the  notion of rational expectation equilibrium  in \citet{Li1993}, 
where  agents' conjectures satisfy a form of reduced game consistency.
\citet{Sasaki1996} show that the rational expectation solution for matching with externalities may fail to exist. \citet{Hafalir2008} introduces a weaker notion of rational expectations, namely sophisticated expectations,  which always  includes the rational expectation solution.
Due to the arbitrary nature of the expectation rules in rational expectations and sophisticated expectations, they  bear no relation to the solution concepts we propose here.} 

Matching with externalities has been studied in several specific economic environments. One strand of the literature focuses on peer effects and preferences over colleagues, whereby agents' preferences depend on the composition of their assigned groups (\citet{DuttaMassò,Fede2007,Pycia2012}). Another concerns matching with couples in labor markets, where partners have joint preferences over pairs of job assignments (\citet{Kojima2013}). Further applications include matching with aggregate externalities, such as those arising from knowledge and public goods (\citet{Fisher}), and competition between matched pairs, as in oligopolistic joint ventures (\citet{kenzo2023}). 

Recently, \citet{Marzena2022,Marzena2023,Pycia2023} study many-to-many matching models with contracts and externalities. 
\citet{Pycia2023} introduce a stability notion in which choice functions, rather than preferences, serve as the primitive objects. Although they mainly focus on P-stability, their results apply to a broader class of stability notions by appropriately modifying the choice function. However, the way in which preferences can be inferred from choice functions, and consequently the conjectures that agents must consider, is not explored in depth.
In contrast, we take preferences as primitive and explicitly investigate how agents form conjectures given a preference profile. Furthermore, unlike \citet{Pycia2023}, who rely on a substitutability condition that rules out complementarities in agents’ preferences even in the one-to-one case,\footnote{For instance, the framework in \citet{Pycia2023} cannot generate predictions for married couples who search for jobs in the same geographic area.} our existence results do not depend on any such sufficient condition.\footnote{On the role of complementarities and peer effects in matching and coalition formation, and on conditions that accommodate them while preserving the existence of stable outcomes, see \citet{Pycia2012} and \citet{Rostek2020}.}
\citet{Marzena2022,Marzena2023} introduce a stability notion that allows for predictions in a wide range of matching environments, including the roommate problem. Their notion also employs conjectures, but only to restrict the set of possible pairwise blockings, rather than, as in our approach, to represent agents’ beliefs about how the market would rearrange if they were to match in a given way. 
Unlike our contribution,  the stability notion of \citet{Marzena2022,Marzena2023} does not generalize the  canonical stability notion \citep{Gale1962}.\footnote{In the one-to-one matching model without externalities it does not  reduce to stability \citep{Gale1962}, but strictly enlarges it.}
\medskip

 Our contribution aligns with what we term  “Pearce–Bernheim program,” a line of research  that   aims to develop
solutions for games in reduced-form, and provide  epistemic foundations, by   adapting  (pre-)epistemic  notions from non-cooperative game theory.\footnote{We distinguish this line of research from the Kreps-Wilson program, as put forward by \citet{LiuQ2020,LiuQ2023,LiuQFor}, who seeks to  develop solutions for cooperative games with incomplete information by  addressing the question of which outcomes can be stable by pinning down the exact belief that supports stable outcomes.}

Within this research program, \citet{Liu2014} make a pioneering
contribution to the problem of defining plausible beliefs in static matching
markets with incomplete information.\footnote{See also \citet{Chen2023}.}
Their stability notion is defined through an iterative procedure that restricts
the beliefs firms may hold about workers' types. Building on this approach,
\citet{Wang2023} introduces rationalizable stability and provides an epistemic
foundation. \citet{Pomatto2022} provides a different foundation for the same
stability notion. He models pairwise deviations as explicit games of offers and
uses forward-induction reasoning.
In contrast,  we employ a conjectural stability approach in a complete-information environment with externalities.
Uncertainty is counterfactual: after a deviation, agents must conjecture how the
residual market would be matched. We discipline these conjectures by
rationalizability while keeping deviations in reduced form. Unlike
\citet{Pomatto2022}, we do not specify an extensive-form game of offers.

\citet{Pour2024} apply related iterative reasoning to an assignment game with
privately observed payments. The assignment and the surplus generated by each
worker--firm pair are common knowledge, but the division of surplus within a
matched pair is observed only by that pair. Thus, their iterative procedure disciplines inferences about unobserved
surplus division. Our procedure instead disciplines conjectures about how the
residual market would be matched after a counterfactual deviation.

\citet{Mauleon2021} and \citet{Mauleon2025} study notions of rationalizable conjectural stability in, respectively, networks and school-choice problems with imperfect information and feedback about others' links or matches. Their approach builds on the rationalizable conjectural equilibrium framework of \citet{Rub1994}. Our notion is different: it builds on the conjectural-stability tradition of \citet{Sasaki1996}, and rationalizability  is used to discipline counterfactual conjectures.

A further strand of rationalizability-based work studies coalitional and social environments, including \citet{Herings2004} and \citet{Ambrus2006}.
While related in spirit, these papers do not address our question: how to endogenize the partner-dependent conjectures that make pairwise deviations well defined under externalities.

\subsection*{A motivating example}
\label{preview}

\begin{example}[A Tale of Three Brothers]
\label{ex}

Consider a two-sided market with three brothers $a_1,a_2,a_3$ on side $A$ and one 
employer $b_1$ on side $B$. 
The brother $a_1$ is the oldest and most experienced, $a_2$ is the middle brother, 
and $a_3$ is the youngest. 
The employer has one position to fill. 
There are four possible matchings: no brother is hired, or exactly one of the three 
brothers is hired. 
We denote by $\mu^0$ the matching in which no one is hired, and by $\mu^i$ the 
matching in which $a_i$ is hired, for each $i=1,2,3$. 
These matchings are displayed in \hyperref[exFig1]{\textcolor{Gold1}{Figure 1}}.

The employer values experience: her most preferred outcome is to hire $a_1$, followed 
by $a_2$. 
Hiring $a_3$, however, is worse for her than leaving the position vacant.

Each brother cares not only about whether he is hired, but also about
whether, and which of, his siblings is hired. Thus, even when he remains
unmatched, his welfare depends on the matching of the other agents,
giving rise to externalities. 
For instance, the oldest brother's preferences over matchings rank the matching in which no one is employed first, followed by the matching in which he is employed, the matching in which the middle brother is employed, and finally the matching in which the youngest brother is employed.
The complete preference profile is on the left-hand side of 
\hyperref[exFig2]{\textcolor{Gold1}{Figure 2}}.

\begin{figure}[ht]
    \centering
\begin{tikzpicture}[scale=0.75, transform shape,
    box/.style={draw,inner sep=7pt,rounded corners=5pt}]
    \node at (1, 0.2) {\Large$\mu^0$};
    \node[draw, circle] (a1) at (0, -1) {$a_1$};
    \node[draw, circle] (a2) at (0, -2) {$a_2$};
    \node[draw, circle] (a3) at (0, -3) {$a_3$};
    \node[draw, circle] (b1) at (2, -2) {$b_1$};
    \node[box,fit=(a1)(a2)(a3)(b1)] {};
\end{tikzpicture}
\hspace{0.2cm}
\begin{tikzpicture}[scale=0.75, transform shape,
    box/.style={draw,mycolor,inner sep=7pt,rounded corners=5pt}]
    \node at (1, 0.2) {\Large$\mu^1$};
    \node[draw, circle] (a1) at (0, -1) {$a_1$};
    \node[draw, circle] (a2) at (0, -2) {$a_2$};
    \node[draw, circle] (a3) at (0, -3) {$a_3$};
    \node[draw, circle] (b1) at (2, -2) {$b_1$};
    \draw (a1) -- (b1);
    \node[box,fit=(a1)(a2)(a3)(b1)] {};
\end{tikzpicture}
\hspace{0.2cm}
\begin{tikzpicture}[scale=0.75, transform shape,
    box/.style={draw,mycolor,inner sep=7pt,rounded corners=5pt}]
    \node at (1, 0.2) {\Large$\mu^2$};
    \node[draw, circle] (a1) at (0, -1) {$a_1$};
    \node[draw, circle] (a2) at (0, -2) {$a_2$};
    \node[draw, circle] (a3) at (0, -3) {$a_3$};
    \node[draw, circle] (b1) at (2, -2) {$b_1$};
    \draw (a2) -- (b1);
    \node[box,fit=(a1)(a2)(a3)(b1)] {};
\end{tikzpicture}
\hspace{0.2cm}
\begin{tikzpicture}[scale=0.75, transform shape,
    box/.style={draw,mycolor,inner sep=7pt,rounded corners=5pt}]
    \node at (1, 0.2) {\Large$\mu^3$};
    \node[draw, circle] (a1) at (0, -1) {$a_1$};
    \node[draw, circle] (a2) at (0, -2) {$a_2$};
    \node[draw, circle] (a3) at (0, -3) {$a_3$};
    \node[draw, circle] (b1) at (2, -2) {$b_1$};
    \draw (a3) -- (b1);
    \node[box,fit=(a1)(a2)(a3)(b1)] {};
\end{tikzpicture}
    \caption{Matchings among $a_1,a_2,a_3$, and $b_1$. An edge between $a_i$ and 
    $b_1$ indicates that the two agents are matched.}
    \label{exFig1}
\end{figure}

\begin{figure}[ht]
\hspace{-0.5cm}
  \begin{minipage}[t]{0.35\textwidth}
      \begin{tabular}[c]{c|c|c|c}
\multicolumn{3}{c|}{$A$} & \multicolumn{1}{|c}{$B$} \\\hline
$a_1$ & $a_2$ & $a_3$ & $b_1$  \\\hline
\rule{0pt}{1em} 
$\mu^0$ & $\mu^3$ & $\mu^1$ & $\mu^1$   \\
$\mu^1$ & $\mu^2$ & $\mu^2$ & $\mu^2$  \\
$\mu^2$ & $\mu^1$ & $\mu^3$ & $\mu^0$  \\
$\mu^3$ & $\mu^0$ & $\mu^0$ & $\mu^3$
\end{tabular}
  \end{minipage}\hspace{-1cm}
    \begin{minipage}[t]{0.60\textwidth}
\begin{tabular}{c|c|c|c}
       & \( a_1 \) & \( a_2 \) & \( a_3 \) \\
\hline
\rule{0pt}{1em} 
{\footnotesize \texttt{Step 0}} 
& $\mu^0_{-a_1},\mu^2_{-a_1},\mu^3_{-a_1}$       
& $\mu^0_{-a_2},\mu^1_{-a_2},\mu^3_{-a_2}$          
& $\mu^0_{-a_3},\mu^1_{-a_3},\mu^2_{-a_3}$         \\
\rule{0pt}{1em} 
{\footnotesize\texttt{Step 1}} 
& $\mu^0_{-a_1},\mu^2_{-a_1}$       
& $\mu^0_{-a_2},\mu^1_{-a_2}$       
& $\mu^0_{-a_3},\mu^1_{-a_3},\mu^2_{-a_3}$         \\
\rule{0pt}{1em} 
{\footnotesize\texttt{Step 2}}
& $\mu^2_{-a_1}$     
& $\mu^0_{-a_2},\mu^1_{-a_2}$        
& $\mu^1_{-a_3},\mu^2_{-a_3}$    \\ 
\rule{0pt}{1em} 
{\footnotesize\texttt{Step 3}}
& $\mu^2_{-a_1}$     
& $\mu^1_{-a_2}$        
& $\mu^1_{-a_3}$    \\ 
\end{tabular}
    \end{minipage}
    \caption{Left: agents' preferences, listed from most to least preferred. 
    Right: iterated elimination of inconceivable conjectures. Each entry is a 
    set of residual matchings; for instance, $\mu^2_{-a_1}$ denotes the 
    restriction of $\mu^2$ to the agents other than $a_1$.}
    \label{exFig2}
\end{figure}

This market has no P-stable matching: $\mu^1$ is blocked by $a_1$, $\mu^0$ is 
blocked by $(a_2,b_1)$, $\mu^2$ is blocked by $(a_1,b_1)$, and $\mu^3$ is blocked 
by $b_1$. 
These blocking arguments, however, all rely on specific assumptions about the 
conjectures held by the deviating agents. 
Before proceeding further, it is useful to clarify what agents conjecture. A conjecture is a residual matching. Thus, if \(a_1\) contemplates remaining single, he conjectures how the agents in \(N\setminus\{a_1\}\) are matched.

Consider for instance the matching $\mu^1$. This matching is the preferred one by the employer who has no incentive to deviate. Agent $a_1$ has incentive to deviate because according to P-stability he holds the conjecture that if he deviates from $\mu^1$ the resulting matching is $\mu^0_{-a_1}$. However, if \(a_1\) could conjecture \(\mu^2_{-a_1}\) after leaving \(b_1\), the deviation from \(\mu^1\) would be deterred. It is precisely this assumption about agents' conjectures under P-stability that the following elimination procedure calls into question.
\medskip

\noindent{\bf Iterated Elimination.} In a way that is reminiscent of the iterated elimination of dominated strategies, we iteratively eliminate conjectures  that no agent can rationally  hold.\footnote{In \Cref{foundation} we discuss in detail which  epistemic assumptions motivate this procedure.}
Conjectures are trivial for the agent $b_1$ and for the agent on side $A$ who is  matched with $b_1$: the 
residual matching is then uniquely determined, since the remaining agents on side $A$ 
are all unmatched. 
The only non-trivial conjectures are those of unmatched agents on side $A$, each of 
whom must form a belief about whether $b_1$ remains single or matches one of the 
other brothers.

\noindent{\bf Step 0.}
Each single agent $a_i$ begins with the full set of possible conjectures. 
For instance, $a_1$'s initial conjecture set is 
$\left\{\mu^0_{-a_1},\mu^2_{-a_1},\mu^3_{-a_1}\right\}$.

\noindent{\bf Step 1.}
Since hiring $a_3$ is not individually rational for $b_1$, neither $a_1$ nor $a_2$ 
can rationally believe that $b_1$ would match with $a_3$ when they are not matched with $b_1$. 
The conjecture $\mu^3_{-a_1}$ is therefore eliminated from $a_1$'s set, and 
$\mu^3_{-a_2}$ from $a_2$'s set.

\noindent{\bf Step 2.}
Since $\mu^3_{-a_2}$ has been eliminated, $a_2$ would consider remaining single only 
if he believed $b_1$ would match with $a_3$, a conjecture no longer available. 
Moreover, $b_1$ strictly prefers $a_2$ over remaining unmatched. 
Together, these facts imply that $a_1$ cannot rationally believe $b_1$ would remain 
single: if $a_1$ is unmatched, he should believe that $a_2$ and $b_1$ will match. 
The conjecture $\mu^0_{-a_1}$ is therefore eliminated, leaving $a_1$ with the unique 
conjecture $\mu^2_{-a_1}$. For the same reason $a_3$ cannot believe $b_1$ would remain single, because both $b_1$ and $a_2$ (given his conjecture) would prefer to match. Therefore, 
$\mu^0_{-a_3}$ is eliminated too.

\noindent{\bf Step 3.}
Since $a_1$ prefers matching with $b_1$ to remaining alone both $a_2$ and $a_3$, when single, must conjecture that $b_1$ is 
matched with $a_1$: the unique remaining conjecture for each is $\mu^1_{-a_2}$ and 
$\mu^1_{-a_3}$, respectively.

The steps of this procedure are summarized on the right-hand side of 
\hyperref[exFig2]{\textcolor{Gold1}{Figure 2}}.
These surviving conjectures support the stability of $\mu^1$: given that $a_1$ 
believes deviating leads to $\mu^2$, which he likes less than $\mu^1$, he has no 
incentive to leave $b_1$; and $b_1$ prefers $\mu^1$ to any deviation. 
The matching $\mu^1$ is  stable when agents' conjectures are \lq\lq rationalizable". In the rest of the paper we formalize this intuition.

\end{example}

\section{Setup}

Let $A$ and $B$ be two mutually disjoint finite {\bf sets of agents} and let $ N \equiv A\cup B$. A coalition $S$ is any non-empty subset of $  N$. 
A  {\bf matching} is a bijection $\mu:  N\longrightarrow   N$ such that 
(i) $\mu(\mu(i))=i$ for all $i\in   N$ and (ii) $\mu(a)\in B\cup\{a\}$, and $\mu(b)\in A\cup \{b\}$ for all $a\in A$ and $b\in B$.
Let $\mathcal{M}$ be the set of all matchings. 
 
Each agent $i \in N$ has a complete and transitive {\bf preference ordering} 
$\succsim_i$ over $\mathcal{M}$. Preferences are defined over entire matchings, 
not just over partners. This is the most general way to capture externalities. 
We write $\succ_i$ for the strict part and $\sim_i$ for the 
indifference part. A preference profile is 
$\succsim\, \equiv (\succsim_i)_{i \in N}$.
A {\bf matching game with externalities} is thus $\Gamma:=\langle A,B,\succsim\rangle$.

\paragraph*{Partners and residual matchings.}
For any $\mu \in \mathcal{M}$ and $i\in N$, we call $\mu(i)$ the partner of $i$, and we write 
$\mathcal P_i$ for the set of all possible partners of $i$, including $i$ herself. Fix 
$i \in N$ and $j \in \mathcal  P_i$. Let $\mathcal{M}_{-i,j}$ be the set of matchings of 
$N \setminus \{i,j\}$.\footnote{Formally, each $\mu_{-i,j} \in 
\mathcal{M}_{-i,j}$ is a bijection $\mu_{-i,j} : N\setminus\{i,j\} \to 
N\setminus\{i,j\}$ such that $\mu_{-i,j}(\mu_{-i,j}(k)) = k$ for all 
$k \in N\setminus\{i,j\}$, with $\mu_{-i,j}(a) \in (B\setminus\{i,j\})\cup\{a\}$ 
for all $a \in A\setminus\{i,j\}$ and $\mu_{-i,j}(b) \in 
(A\setminus\{i,j\})\cup\{b\}$ for all $b \in B\setminus\{i,j\}$.} If 
$\mu(i) = j$, we write $\mu_{-i,j} \in \mathcal{M}_{-i,j}$ for the restriction 
of $\mu$ to $N \setminus \{i,j\}$.
\medskip

\noindent Given \(\nu_{-i,j}\in\mathcal M_{-i,j}\), define
\(\nu_{-i,j}+(i,j)\) as the unique matching \(\mu\in\mathcal M\) such that
$
\mu(k)=\nu_{-i,j}(k)\quad\text{for every }k\in N\setminus\{i,j\},
$
$
\mu(i)=j$ and  $\mu(j)=i.
$
When \(i=j\), we write \(\nu_{-i}+(i)\) instead of
\(\nu_{-i,i}+(i,i)\).

\paragraph*{Conjectures.}
A conjecture of $i$ given partner $j \in \mathcal P_i$ is any element 
$\mu_{-i,j} \in \mathcal{M}_{-i,j}$.

 \begin{definition}[System of Conjectures]
 \label{conjectures}  Define $\mathscr{P}\equiv\{(i,j)\in N\times N|j\in \mathcal{P}_i\}$.
      A  system of conjectures is a correspondence $\mathbf{C}(\, \cdot \, ,\succsim):   \mathscr{P}\rightrightarrows \bigcup_{(i,j)\in\mathscr{P} }\mathcal{M}_{-i,j}$ such that for all $(i,j)\in N\times N$ with $j\in \mathcal{P}_i$ it holds that  $\mathbf{C}(i,j,\succsim) \subseteq  \mathcal{M}_{-i,j}$. A  system of conjectures is non-degenerate if for all $(i,j)\in N\times N$ with $j\in \mathcal{P}_i$, $\mathbf{C}(i,j,\succsim) \neq \emptyset$. 
       \end{definition}
  Then,  given a preference profile $\succsim$, an agent $i\in N$ and a partner $j\in \mathcal{P}_i$, 
 $\mathbf{C}(i,j,\succsim) \subseteq  \mathcal{M}_{-i,j}$  denotes 
the set of possible conjectures of $i$ when she is paired with $j$.
When $i=j$ we simply write $\mathcal{M}_{-i}$ and  
$\mathbf{C}(i,\succsim)\equiv \mathbf{C}(i,i,\succsim) \subseteq  \mathcal{M}_{-i}$.

\begin{remark}
  All the solution concepts  studied in this paper build on    symmetric conjecture systems. For notational simplicity we therefore
assume throughout that
$
\mathbf C(a,b,\succsim)=\mathbf C(b,a,\succsim)
$
for every  pair $(a,b)\in A \times B$.
\end{remark}

\paragraph*{Conjectural Stability.} The next  definition introduces  a class of stability notion, conjectural stability or  $\mathbf{C}$-stability in short,
in which an agent deviates only if she improves under 
every conjecture she holds.

\begin{definition}{\bf (Conjectural Stability).\footnote{
    Conjectural stability presented here differs from the notion introduced by \citet{Sasaki1996}. First, we do not include the matching of conjecturer agents as a part of their conjectures. Then, in \citet{Sasaki1996},   singles are not permitted and Condition 1 is  replaced by a  notion requiring  that any conjecturally stable matching $\mu$ is  conjectured by all pairs that are matched according to $\mu$.  Notably, we will show (see \Cref{admiStable}) that the solutions we propose in this paper satisfy this requirement.}}
\label{stability}
Fix a matching game with externalities  $\Gamma:=\langle A,B,\succsim\rangle$ and a system of conjectures $\mathbf{C}$. A matching $\mu$ is  $\mathbf{C}$-stable if the following conditions hold:
\begin{enumerate}
 \item $\mathbf{C}$-individual rationality:
  $\nexists  k\in   N$  such that $\nu_{-k}+(k)\succ_k \mu $
    $\forall \nu_{-k}\in \mathbf{C}(k,\succsim)$

     \item  $\mathbf{C}$-unblocking:
    $\nexists(a,b)\in A\times B$  such that  $\nu_{-a,b}+(a,b)\succ_{a}\mu$ and $\nu_{-a,b}+(a,b)\succ_b \mu$
$\forall \nu_{-a,b}\in \mathbf{C}(a,b,\succsim)$

     \end{enumerate}
\end{definition}

\begin{remark}[Gale--Shapley Stability]
For every non-degenerate conjecture system $\mathbf C$, if preferences exhibit no externalities, then a matching $\mu\in\mathcal M$ is $\mathbf C$-stable if and only if it is Gale-Shapley stable.\footnote{Preferences exhibit no externalities if, for every $k\in N$ and every $\mu,\nu\in\mathcal M$,
$
\mu(k)=\nu(k)\quad\Longrightarrow\quad \mu\sim_k\nu.
$}
Thus, $\mathbf C$-stability generalizes Gale--Shapley stability.
\end{remark}

P-stability \citep{Sasaki1986}, the predominant notion in the literature
\citep[e.g.,][]{Roth1990}, is formally defined below.

\begin{definition}[P-stability]
\label{P-stability}
 $\mu$ is P-stable if it is
$\mathbf C^P_\mu$-stable where $\mathbf C^P_\mu$ is the system of conjectures given by:
\begin{enumerate}
 \item For $k\in N$,
 $\mathbf C^P_\mu(k,\succsim)=\{\mu^P_{-k}\}$, where
 $\mu^P_{-k}(\mu(k))=\mu(k)$ if $\mu(k) \in N \setminus \{k\}$ and $\mu^P_{-k}(j)=\mu(j)$ for all
 $j\notin\{k,\mu(k)\}$.

 \item For $(a,b)\in A\times B$,
 $\mathbf C^P_\mu(a,b,\succsim)=\{\mu^P_{-a,b}\}$, where
 $\mu^P_{-a,b}(\mu(a))=\mu(a)$ if $\mu(a) \in N \setminus \{a,b\}$, $\mu^P_{-a,b}(\mu(b))=\mu(b)$ if $\mu(b) \in N \setminus \{a,b\}$, and
 $\mu^P_{-a,b}(j)=\mu(j)$ for all $j\notin\{a,b,\mu(a),\mu(b)\}$.\footnote{When
 the deviating agent is already single under $\mu$ (that is, $\mu(k)=k$,
 $\mu(a)=a$, or $\mu(b)=b$), the condition on the corresponding former partner is
 vacuous, since that partner coincides with the deviating agent and is removed
 from the residual market.}
\end{enumerate}
\end{definition}

 The next proposition records the basic comparative logic of the framework: 
enlarging conjectures enlarges the set of stable matchings.  It explains why 
conjectures must be disciplined to obtain sharp predictions.
 All proofs are collected in the \hyperlink{Appendix}{\sc Appendix}. 

 \begin{proposition}
 \label{nestedness}
Let $\mathbf{C},\mathbf{C}'$ be two systems of conjectures such that 
$\mathbf{C}(i,j,\succsim)\subseteq \mathbf{C}'(i,j,\succsim)$ for all 
$i,j\in N$ with $j\in\mathcal{P}_i$. Then every 
$\mathbf{C}$-stable matching is $\mathbf{C}'$-stable.
 \end{proposition}

\section{Conjecture-Rationalizable Stability}
\label{rationalizableconjectures}

One of the most important  aspects in a theory of conjectural stability for matching games with externalities is  determining agents' conjectures. This section offers a systematic treatment of the problem for this purpose.

We introduce the notion of 
{\bf C}-dominated conjectures motivated by the idea of dominated strategies in non-cooperative game theory. The guiding idea is simple:  conjecture $\mu_{-i,j}$ is $\mathbf{C}$-dominated if $\mu_{-i,j}$ prescribes a dominated behavior for some agent or pair in 
$N\setminus\{i,j\}$. 

To see the analogy with non-cooperative games, fix a system of conjectures $\mathbf{C}$ and an   agent $i$  holding a conjecture under which   some other  
agent $k$ is matched to $\mu(k)$.
 Think of \lq\lq $k$ matches $\mu(k)$'' as an action available to $k$, alongside the action 
\lq\lq $k$ stays single".
  Under this interpretation, \lq\lq $k$ matches $\mu(k)$'' is 
{\bf C}-dominated by \lq\lq $k$ stays single" if the latter is strictly preferred to the former  by  $k$  under all conjectures that $k$ may hold. When it is the case,  $i$'s conjecture is itself $\mathbf{C}$-dominated. 
Condition (2) below has no such counterpart. Two agents $a$ and $b$ who
are not matched to each other eliminate a conjecture only if \emph{both}
strictly prefer matching with each other, regardless of the conjectures
they hold. Elimination therefore relies on joint rather than individual
rationality. This is the cooperative ingredient of our approach. It is
the same ingredient that underlies pairwise stability in
\citet{Gale1962}, and it plays an analogous role in the procedures of
\citet{Liu2014} and \citet{Wang2023}.

\begin{definition}[$\mathbf C$-Dominated Conjectures]
\label{conjecturedominance}
Fix $\Gamma=\langle A,B,\succsim\rangle$ and a system of conjectures $\mathbf C$.
Let $i\in N$, $j\in\mathcal P_i$, and $\mu_{-i,j}\in\mathcal M_{-i,j}$.
For notational simplicity, write $\mu(k)$ for $\mu_{-i,j}(k)$.
The conjecture $\mu_{-i,j}$ is \textbf{$\mathbf C$-dominated} if one of the following conditions holds:

\begin{description}
\item[(1)] There exists $k\in N\setminus\{i,j\}$ such that
$\nu_{-k}+(k)\succ_k \nu'_{-k,\mu(k)}+(k,\mu(k))$
for every $\nu_{-k}\in\mathbf C(k,\succsim)$ and every
$\nu'_{-k,\mu(k)}\in\mathbf C(k,\mu(k),\succsim)$.

\item[(2)] There exists $(a,b)\in(A\setminus\{i,j\})\times(B\setminus\{i,j\})$
such that
$\nu_{-a,b}+(a,b)\succ_a \nu'_{-a,\mu(a)}+(a,\mu(a))$
and
$\nu_{-a,b}+(a,b)\succ_b \nu''_{-b,\mu(b)}+(b,\mu(b))$
for every $\nu_{-a,b}\in\mathbf C(a,b,\succsim)$,
every $\nu'_{-a,\mu(a)}\in\mathbf C(a,\mu(a),\succsim)$, and every
$\nu''_{-b,\mu(b)}\in\mathbf C(b,\mu(b),\succsim)$.
\end{description}

A conjecture is \textbf{$\mathbf C$-undominated} if it is not $\mathbf C$-dominated.
\end{definition}

The concept of $\mathbf{C}$-dominance  allows us to   formalize  a  rationalizable procedure.

 \begin{definition}[Rationalizable procedure]
    \label{procedure}
   Fix $\Gamma:=\langle A,B, \succsim \rangle$ and consider the following procedure.
\begin{description}
        \item ({\tt Step $n=0$}) For each $(i,j)\in N\times N$ with $j\in \mathcal{P}_i$ let  $\mathbf{C}^0(i,j,\succsim)=\mathcal{M}_{-i,j}$
        \item ({\tt Step $n>0$})  Assume $\mathbf{C}^{n-1}$ has been defined and   for  each $(i,j)\in N\times N$ with $j\in \mathcal{P}_i$ let 
        $  \mathbf{C}^n(i,j,\succsim)=\{\mu_{-i,j}\in  \mathcal{M}_{-i,j}|\mu_{-i,j} \ \text{is} \ \mathbf{C}^{n-1}\text{-undominated}\}$
    \end{description}
    \noindent
    Let \(\mathbf C^\infty\) be the system of {\bf rationalizable conjectures} defined by
$
\mathbf C^\infty(i,j,\succsim)
=
\bigcap_{n\in\mathbb N}\mathbf C^n(i,j,\succsim)
$
for every  \((i,j)\in\mathscr P\).

\end{definition}

We call the set of $\mathbf{C}^\infty$-stable matchings the set of {\bf conjecture-rationalizable stable matchings} and denote it $\mathbf{S}^\infty$.

The following example provides a simple application of the concepts of rationalizable conjectures and conjecture-rationalizable stable matchings. 
  \begin{example}[The Two Gentlemen of Verona]
\label{esempioTGoV}

Consider the final scene of Shakespeare's \emph{Two Gentlemen of Verona}, discussed in the introduction. 
Let $a_1$ be Valentine, $a_2$ be Proteus, and $b_1$ be Silvia. 
There are three possible matchings: no one marries Silvia, Valentine marries Silvia, or Proteus marries Silvia. 
We denote them by $\mu^0$, $\mu^1$, and $\mu^2$, respectively. 
 \hyperref[exA]{\textcolor{Gold1}{Figure 3}} displays them together with preferences.

\begin{figure}[ht]
\centering

\begin{tabular}{c c c c}

\begin{tikzpicture}[scale=0.75, transform shape,
    box/.style={draw,inner sep=7pt,rounded corners=5pt}]
    \node at (1, 0.2) {\Large$\mu^0$};

    \node[draw, circle] (a1) at (0, -1) {$a_1$};
    \node[draw, circle] (a2) at (0, -2) {$a_2$};
    \node[draw, circle] (b1) at (2, -2) {$b_1$};

    \node[box,fit=(a1)(a2)(b1)] {};
\end{tikzpicture}

&

\begin{tikzpicture}[scale=0.75, transform shape,
    box/.style={draw,mycolor,inner sep=7pt,rounded corners=5pt}]
    \node at (1, 0.2) {\Large$\mu^1$};

    \node[draw, circle] (a1) at (0, -1) {$a_1$};
    \node[draw, circle] (a2) at (0, -2) {$a_2$};
    \node[draw, circle] (b1) at (2, -2) {$b_1$};

    \draw (a1) -- (b1);
    \node[box,fit=(a1)(a2)(b1)] {};
\end{tikzpicture}

&

\begin{tikzpicture}[scale=0.75, transform shape,
    box/.style={draw,mycolor,inner sep=7pt,rounded corners=5pt}]
    \node at (1, 0.2) {\Large$\mu^2$};

    \node[draw, circle] (a1) at (0, -1) {$a_1$};
    \node[draw, circle] (a2) at (0, -2) {$a_2$};
    \node[draw, circle] (b1) at (2, -2) {$b_1$};

    \draw (a2) -- (b1);
    \node[box,fit=(a1)(a2)(b1)] {};
\end{tikzpicture}

&

\raisebox{1.2cm}{
\begin{tabular}[c]{c|c|c}
\multicolumn{2}{c|}{$A$} & \multicolumn{1}{|c}{$B$} \\\hline
$a_1$ & $a_2$ & $b_1$ \\\hline
\rule{0pt}{1em}
$\mu^2$ & $\mu^1$ & $\mu^1$ \\
$\mu^1$ & $\mu^2$ & $\mu^2$ \\
$\mu^0$ & $\mu^0$ & $\mu^0$
\end{tabular}
}

\end{tabular}

\caption{Agents' preferences and possible matchings.}
\label{exA}
\end{figure}

We first observe that $\mu^0_{-a_1}$ is $\mathbf C^0$-undominated as a conjecture
for $a_1$ when he is single. That is, $a_1$, when he is single, can
conjecture that also $a_2$ is single. In fact, there is a conjecture
$\mu^1_{-a_2}\in \mathbf C^0(a_2,\succsim)$ such that
$   \mu^1_{-a_2}+(a_2)
    =\mu^1
    \succ_{a_2}
    \mu^2
    =\mu^2_{-a_2,b_1}+(a_2,b_1).
$
Intuitively, when Valentine is not matched with Silvia, he may conjecture
that also Proteus is not matched with her because Proteus may rationally
remain single if he conjectures that Valentine and Silvia marry.

By a symmetric argument, $\mu^0_{-a_2}$ is $\mathbf{C}^{0}$-undominated
as a conjecture for $a_2$, since $\mu^2_{-a_1}\in\mathbf C^{0}(a_1,\succsim)$
and $\mu^2_{-a_1}+(a_1)=\mu^2\succ_{a_1}\mu^1=\mu^1_{-a_1,b_1}+(a_1,b_1)$.  

Iterating similar arguments for all conjectures in $\mathbf{C}^{0}$ it can be proved that  $\mathbf{C}^0=\mathbf{C}^{1}=\mathbf{C}^{\infty}$.

Under conjecture-rationalizable stability, a deviation is admissible only if it is profitable under all rationalizable conjectures of the deviating agents. Deviations from $\mu^1$ and $\mu^2$ are deterred by the rationalizable conjectures leading to $\mu^0$. 
 Both matchings are therefore 
conjecture-rationalizable stable.

\begin{remark}[Rational Expectations]
This example also separates our concept from the rational-expectations solution 
of \citet{Li1993}. Under rational expectations, a deviating pair anticipates that the remaining agents will rearrange  into a stable matching of the residual market.

Fix a matching game with externalities
\(\Gamma=\langle A,B,\succsim\rangle\). For every 
\((i,j)\in N\times N\) with \(j\in\mathcal P_i\), let
\(\Gamma_{-i,j}^{\,i,j}\) denote the reduced matching game obtained after
\((i,j)\)  forms and leaves the market.\footnote{Its set of
agents is \(N\setminus\{i,j\}\), and the preferences of every remaining agent
\(k\in N\setminus\{i,j\}\) are defined by
$
\nu_{-i,j}\succsim^{\,i,j}_{k}\nu'_{-i,j}
\quad\Longleftrightarrow\quad
\nu_{-i,j}+(i,j)\succsim_k \nu'_{-i,j}+(i,j),
$
for all \(\nu_{-i,j},\nu'_{-i,j}\in\mathcal M_{-i,j}\).
When \(i=j\), we write \(\Gamma_{-i}^{\,i}\) and use
\(\nu_{-i}+(i)\).}
A system of conjectures \(\mathbf C^{RE}\) is a
\textbf{rational-expectations system of conjectures} if, for every 
$(i,j)\in N\times N$ with \(j\in\mathcal P_i\),
$
\mathbf C^{RE}(i,j,\succsim)
=
\mathbf S^{RE}\!\left(\Gamma_{-i,j}^{\,i,j}\right),
$
where \(\mathbf S^{RE}\!\left(\Gamma_{-i,j}^{\,i,j}\right)\) is the set of
rational-expectations stable matchings of the reduced game
\(\Gamma_{-i,j}^{\,i,j}\).

In this  example, no stable matching with rational expectations exists. The matching  $\mu^1$  is not stable because Valentine would deviate remaining single: in 
the reduced economy of Silvia and Proteus, the unique stable matching pairs 
them, which is Valentine's best outcome. The same holds for $\mu^2$ by symmetry. 
Rational expectations thus require that, whenever an agent conjectures herself 
single, all others assume she stays single. This grants the conjecturing agent a kind of commitment power that is hard to justify. 
\end{remark}
\end{example}

In contrast, conjecture-rationalizable stable matchings always exist. This is our first main result.

\begin{theorem}[Existence]
\label{th1}
Fix $\Gamma:=\langle A,B, \succsim \rangle$. Then $\mathbf{S}^\infty\neq \emptyset$.
\end{theorem}

We next characterize the system of rationalizable conjectures as a fixed point. 
The notion of self-undominance below mirrors the self-stabilizing property of 
\citet{Liu2014, Wang2023}, itself inspired by Pearce's best-reply set 
\citep{Pearce1984}.

 A system of conjectures $\mathbf{C}$ is self-undominated if every conjecture of every agent is $\mathbf{C}$-undominated.

\begin{definition}[Self-undominance] Fix $\Gamma:=\langle A,B, \succsim  \rangle$. 
    A system of conjectures $\mathbf{C}$ is said  \textbf{self-undominated} if for all $(i,j) \in N\times N$ with $j\in \mathcal{P}_i$ and $\mu_{-i,j}\in \mathcal{M}_{-i,j}$, it holds that  $\mu_{-i,j}\in \mathbf{C}(i,j,\succsim)$ only if it  is $\mathbf{C}$-undominated.
\end{definition}

The following fixed-point characterization shows that the system of rationalizable conjectures  is well-behaved.

\begin{proposition}[Fixed-Point characterization]
\label{th2}
It holds that: 
    \begin{enumerate}
        \item If $\mathbf{C}$ is self-undominated, then $\mathbf{C}(i,j,\succsim)\subseteq \mathbf C^\infty(i,j,\succsim)$  for all $(i,j)\in N\times N$ with  $j\in \mathcal{P}_i$.
        \item The system of rationalizable conjectures $\mathbf{C}^\infty$ is self-undominated. 
    \end{enumerate}
Hence, $\mathbf{C}^\infty$ is the largest self-undominated system of conjectures.
\end{proposition}
\Cref{esempioTGoV} illustrates the fixed-point logic. Although $\mu^0_{-a_1}$ 
and $\mu^0_{-a_2}$ may seem fragile, both survive: each is supported by a 
mutually reinforcing rationalizable conjecture. For instance, when Valentine 
remains single, the conjecture $\mu^0_{-a_1}$ says that Proteus also remains 
single (and Silvia remains unmatched). Valentine cannot rule out this 
conjecture, because Proteus may rationally remain single if he conjectures 
$\mu^1_{-a_2}$, namely that Valentine and Silvia marry. Proteus, in turn, 
cannot rule out $\mu^1_{-a_2}$, because Valentine may rationally marry Silvia 
if he conjectures $\mu^0_{-a_1}$. The two conjectures thus sustain each other: 
$\mu^0_{-a_1}\in \mathbf C^\infty(a_1,\succsim)$ is supported by 
$\mu^1_{-a_2}\in \mathbf C^\infty(a_2,\succsim)$, and conversely. Rationalizable 
conjectures therefore survive not in isolation, but as part of a 
self-undominated system. \hyperref[fig:beliefcycles]{\textcolor{Gold1}{Figure 4}} depicts the two supporting cycles.

\begin{figure}[ht]
\centering

\begin{tikzpicture}[
    conjecture/.style={
        draw,
        rounded corners=5pt,
        inner sep=7pt,
        align=center,
        minimum width=3.4cm,
        minimum height=1.1cm,
        fill=gray!8
    },
    support/.style={->, thick, bend left=22},
    every node/.style={font=\small},
    >=stealth
]

\node[conjecture] (v0) at (0,0) 
{$\mu^0_{-a_1}$\\
{\footnotesize Valentine: no marriage}};

\node[conjecture] (p1) at (0,-2.4) 
{$\mu^1_{-a_2}$\\
{\footnotesize Proteus: Valentine--Silvia}};

\draw[support] (v0) to node[right] {\footnotesize supports} (p1);
\draw[support] (p1) to node[left] {\footnotesize supports} (v0);

\node[align=center] at (0,-3.8) 
{\footnotesize Cycle sustaining $\mu^0_{-a_1}$};

\node[conjecture] (p0) at (5.2,0) 
{$\mu^0_{-a_2}$\\
{\footnotesize Proteus: no marriage}};

\node[conjecture] (v2) at (5.2,-2.4) 
{$\mu^2_{-a_1}$\\
{\footnotesize Valentine: Proteus--Silvia}};

\draw[support] (p0) to node[right] {\footnotesize supports} (v2);
\draw[support] (v2) to node[left] {\footnotesize supports} (p0);

\node[align=center] at (5.2,-3.8) 
{\footnotesize Cycle sustaining $\mu^0_{-a_2}$};

\end{tikzpicture}

\caption{Mutual support of conjectures. The conjectures leading to no marriage survive because each is sustained by another rationalizable conjecture.}
\label{fig:beliefcycles}
\end{figure}

\citet{Sasaki1996} argue that a natural requirement for stability is that the matching under consideration be compatible with agents' conjectures. 
In our terminology, a matching $\mu$ is $\mathbf C$-consistent if each agent regards as possible the residual matching generated by $\mu$. 
Unlike \citet{Sasaki1996}, who build consistency into the definition of stability, we obtain it as a result.\footnote{\citet{Sasaki1996} refer to this requirement as $\mathbf C$-admissibility. Since, in the non-cooperative game theory literature, admissibility is typically associated with undominated behavior, we use the term $\mathbf C$-consistency to avoid confusion. This notion is unrelated to undominance and, in particular, to our notion of $\mathbf C$-dominance.}

\begin{definition}[$\mathbf{C}$-consistency]
A matching \(\mu\) is \(\mathbf C\)-consistent if, for every \(i\in N\),
$
\mu_{-i,\mu(i)}\in \mathbf C(i,\mu(i),\succsim).
$
\end{definition}
We claim that any conjecture-rationalizable stable matching is
\(\mathbf{C}^{\infty}\)-consistent. That is, if \(\mu\) is stable with
rationalizable conjectures, then, for every agent \(i\), the conjecture
\(\mu_{-i,\mu(i)}\) is rationalizable.

\begin{proposition}[$\mathbf{C}^\infty$-consistency]
\label{admiStable}
    Fix $\Gamma:=\langle A,B, \succsim \rangle$. If $\mu\in\mathbf{S}^\infty$, then it is  $\mathbf{C}^\infty$-consistent.
\end{proposition}

We conclude this section by pointing out that not all conjecture-rationalizable stable matchings  are Pareto efficient, but at least one is.\footnote{An example demonstrating the existence of a Pareto dominated conjecture-rationalizable stable matching   is available upon request.}

\begin{proposition}[Pareto efficiency]
\label{Pareto}
    Fix $\Gamma:=\langle A,B, \succsim \rangle$. Then, there is  a conjecture-rationalizable stable  matching that is Pareto efficient.
\end{proposition}

\section{An application: matching with couples}
\label{couple}

One of the most significant shifts in the  labor market over the twentieth century was the growing participation of  women in the workforce, giving rise to many two-career households. Couples searching for two jobs must navigate a difficult coordination process, both between themselves and with prospective employers.

The matching-with-couples problem \citep{Roth1990,Klaus2005} fits naturally within our framework. In that setting, the preferences of one member of a couple may depend on the assignment of the other member.\footnote{Formally, for agents \(s\) and \(s'\), there may exist \(\mu,\mu'\in\mathcal M\) such that
$
\mu(s)=\mu'(s)
\quad\text{but}\quad
\mu\succ_s\mu',
$
because \(\mu\) and \(\mu'\) assign \(s'\) to different partners.}
Our framework allows preferences to depend on the entire matching, with the couples model corresponding to the special case in which this dependence is restricted to the assignment of one’s partner.

It is well established that, in the presence of couples, the set of P-stable matchings may be empty \citep{Roth1990,Roth2008} unless restrictions such as weak responsiveness are imposed on the preference domain \citep{Klaus2005}.\footnote{Informally, a couple’s preferences are responsive if an improvement in one partner's assignment,
according to that partner's individual preferences, also improves the
couple's assignment, provided that the other partner's assignment remains
unchanged. This assumption is often violated in practice, as shown by \citet{Kojima2013}.} We consider a simple example adapted from \citet{Roth2008} that admits no P-stable matching. We show that a    corresponding conjecture-rationalizable stable matching exists and  constitutes  a  sensible prediction.

\begin{figure}[ht]
    \centering
      \begin{tabular}
[c]{c|c|c|c}%
\multicolumn{2}{c|}{$A$} & \multicolumn{2}{|c}{$B$}
\\\hline
$h_1$ & $h_2$ & $(s_1,s_2)$ & $s_3$  \\\hline
\rule{0pt}{1em} 
$s_1$ & $s_3$ & $(h_1,h_2)$  &$h_1$ \\
$s_3$ & $s_2$ & $(\emptyset,h_2)$  & $h_2$\\
$(\emptyset)$ & $(\emptyset)$ & $(\emptyset,\emptyset)$ & $(\emptyset) $

\end{tabular}
\caption{Preferences of hospitals and medical residents. Residents $(s_1,s_2)$ form a couple and have joint preferences over the pairs of jobs they may obtain.}
\label{excouple}
\end{figure}

In this example there are two hospitals offering one job each,  and three medical residents: two medical residents ($s_1,s_2$) form a couple and have aligned preferences
, the two hospitals and $s_3$ do not exhibit externalities in their preferences and they only care about their match (see \hyperref[excouple]{\textcolor{Gold1}{Figure 5}}).
This example can be translated directly into our framework by extending each agent’s ranking over assignments to a preference over matchings. Resident 3 is indifferent among all matchings in which she works for hospital 1, while residents 1 and 2 are indifferent among all matchings that leave their own relevant assignments unchanged, regardless of whether hospital 1 hires resident 3.

This market admits no P-stable matching. Restricting attention to individually rational matchings, every matching in which $s_3$ remains unemployed is blocked by the pair $(h_2,s_3)$. The matching in which $s_3$ is matched with $h_2$ while $h_1$ remains vacant is blocked by $(h_1,s_3)$. The matching in which $s_3$ is matched with $h_1$ while $h_2$ remains vacant is blocked by $(h_2,s_2)$. Finally, the matching in which $s_3$ is matched with $h_1$ and $s_2$ with $h_2$ is blocked by $(h_1,s_1)$. Hence, no P-stable matching exists.

Nevertheless, the matching
$(h_1,s_3;h_2,s_2;\emptyset,s_1)$ is conjecture-rationalizable stable. The potential deviation by $(h_1,s_1)$ is deterred by the rationalizable conjecture
$(h_2,s_3;\emptyset,s_2)\in \mathbf C^\infty(s_1,h_1,\succsim).$ 
Intuitively, resident $s_1$ fears that, if she accepted a position at hospital $h_1$, hospital $h_2$ would  hire $s_3$, leaving her partner $s_2$ unemployed. Since this outcome is less desirable for the couple than the status quo, the deviation is not profitable.

We conclude this section by observing that, even in the domain of (weakly) responsive preferences, the set of P-stable matchings and the set of conjecture-rationalizable stable matchings may be disjoint. To see this, consider a slight modification of the previous example in which the couple $c\equiv(s_1,s_2)$ has the following responsive preferences:
$
(h_1,h_2)\succ_c(\emptyset,h_2)\succ_c(h_1,\emptyset)\succ_c(\emptyset,\emptyset).
$

The matching
$(h_1,s_1; h_2,s_3; \emptyset,s_2)$
is the unique P-stable matching, whereas
$(h_1,s_3; h_2,s_2; \emptyset,s_1)$
is the unique conjecture-rationalizable stable matching.
To check that the former matching is the unique P-stable is immediate. This matching is not  conjecture-rationalizable stable, because the only rationalizable conjecture available to $s_1$ when unemployed is that $s_3$ is employed by $h_1$ while $s_2$ is employed by $h_2$. Under this conjecture, remaining unemployed is strictly better than accepting employment at $h_1$, making the deviation profitable. 

 The matching $(h_1,s_3; h_2,s_2; \emptyset,s_1)$ is conjecture-rationalizable stable because $s_1$ can hold a rationalizable conjecture according to which, after accepting a position at $h_1$, hospital $h_2$ hires $s_3$, leaving $s_2$ unemployed. Consequently, accepting employment at $h_1$ may induce a worse matching for the couple.

\section{Rationalizable Matchings}

\Cref{rationalizableconjectures} has been devoted to studying a stability notion based on rationalizable conjectures. A natural question is whether it is possible to extend the rationalization procedure from the system of conjectures to the set of matchings $\mathcal{M}$. The answer is affirmative, and this extension leads to the  definition of rationalizable matchings.

 We first present an analogous definition for \textbf{C}-dominated matchings.

\begin{definition}[$\mathbf{C}$-Dominated Matching]\footnote{The attentive reader will have noticed that this definition  is similar to the definition of  $\mathbf{C}$-stability  (\Cref{stability}). However, there is a fundamental difference.  In $\mathbf{C}$-stability agents compare all matchings they believe when they deviate from a {\it status quo},  with the {\it status quo} itself,  while  in $\mathbf{C}$-dominance agents compare all the matchings they conjecture when they \lq\lq play an action" with  all the matchings they conjecture when they \lq\lq play a different action".}
\label{dominance}
Fix $\Gamma:=\langle A,B, \succsim \rangle$,  and a system of conjectures $\mathbf{C}$. A matching $\mu\in \mathcal{M}$ is $\mathbf{C}$-dominated if one of the following holds:
    \begin{description}
        \item[(1)] $\exists k\in N$ such that 
         $\nu_{-k}+(k)\succ_k \nu'_{-k,\mu(k)}+(k,\mu(k)) $,
          $\forall \nu_{-k}\in \mathbf{C}(k,\succsim)$ and   $\forall \nu'_{-k,\mu(k)}\in \mathbf{C}(k,\mu(k),\succsim)$. 
    \item[(2)]  $\exists(a,b) \in A \times B$ such that  $\nu_{-a,b}+(a,b)\succ_a \nu'_{-a,\mu(a)}+(a,\mu(a))$  and $\nu_{-a,b}+(a,b)\succ_b \nu''_{-b,\mu(b)}+ (b,\mu(b))$, $\forall \nu_{-a,b}\in \mathbf{C}(a,b,\succsim)$,   $\forall\nu'_{-a,\mu(a)}\in \mathbf{C}(a,\mu(a),\succsim)$ and  $\forall\nu''_{-b,\mu(b)}\in \mathbf{C}(b,\mu(b),\succsim)$.
    \end{description}
\end{definition}

The previous definition naturally leads to the notion of rationalizable matchings.

 \begin{definition}[Rationalizable Matchings] 
 \label{rat}
    Fix  $\Gamma:=\langle A, B,\succsim\rangle$. A matching $\mu$ is rationalizable if it is $\mathbf{C}^\infty$-undominated.
Let $\mathbf{M}^\infty$ denote the set of rationalizable matchings.
\end{definition}

The conceptual distinction between conjecture-rationalizable stability and rationalizability for matchings hinges on the following point. Conjecture-rationalizable stability implicitly assumes that all agents are aware of the putative matching, whereas rationalizability requires only that each agent knows the identity of their own partner.\footnote{This distinction will be explored further in \Cref{foundation}.} Rationalizability for matching is indeed a non-equilibrium notion.

\begin{theorem}[Refinement]
\label{thRef}
Fix $\Gamma:=\langle A,B, \succsim \rangle$. Then $\mathbf{S}^\infty\subseteq \mathbf{M}^\infty$.
\end{theorem}
The following  corollary is implied by  \Cref{th1} and \Cref{thRef}.

\begin{corollary}[Existence]
\label{thRat}
Fix $\Gamma:=\langle A,B, \succsim \rangle$. Then $\mathbf{M}^\infty\neq \emptyset$
\end{corollary}

It is worth noting that rationalizability and  P-stability  are unrelated concepts in matching with externalities.

The following example shows that P-stable  matchings are not necessarily  rationalizable.
\begin{example}
\label{ex3}
Consider a modified version of \Cref{esempioTGoV} in which the preferences of $a_2$ align with those of $a_1$. Matchings and preferences are depicted in \hyperref[exB]{\textcolor{Gold1}{Figure 6}}.

\begin{figure}[ht]
\centering

\begin{tabular}{c c c c}

\begin{tikzpicture}[scale=0.75, transform shape,
    box/.style={draw,inner sep=7pt,rounded corners=5pt}]
    \node at (1, 0.2) {\Large$\mu^0$};

    \node[draw, circle] (a1) at (0, -1) {$a_1$};
    \node[draw, circle] (a2) at (0, -2) {$a_2$};
    \node[draw, circle] (b1) at (2, -2) {$b_1$};

    \node[box,fit=(a1)(a2)(b1)] {};
\end{tikzpicture}

&

\begin{tikzpicture}[scale=0.75, transform shape,
    box/.style={draw,mycolor,inner sep=7pt,rounded corners=5pt}]
    \node at (1, 0.2) {\Large$\mu^1$};

    \node[draw, circle] (a1) at (0, -1) {$a_1$};
    \node[draw, circle] (a2) at (0, -2) {$a_2$};
    \node[draw, circle] (b1) at (2, -2) {$b_1$};

    \draw (a1) -- (b1);
    \node[box,fit=(a1)(a2)(b1)] {};
\end{tikzpicture}

&

\begin{tikzpicture}[scale=0.75, transform shape,
    box/.style={draw,mycolor,inner sep=7pt,rounded corners=5pt}]
    \node at (1, 0.2) {\Large$\mu^2$};

    \node[draw, circle] (a1) at (0, -1) {$a_1$};
    \node[draw, circle] (a2) at (0, -2) {$a_2$};
    \node[draw, circle] (b1) at (2, -2) {$b_1$};

    \draw (a2) -- (b1);
    \node[box,fit=(a1)(a2)(b1)] {};
\end{tikzpicture}

&

\raisebox{1.2cm}{
\begin{tabular}[c]{c|c|c}
\multicolumn{2}{c|}{$A$} & \multicolumn{1}{|c}{$B$} \\\hline
$a_1$ & $a_2$ & $b_1$ \\\hline
\rule{0pt}{1em}
$\mu^2$ & $\mu^2$ & $\mu^1$ \\
$\mu^1$ & $\mu^1$ & $\mu^2$ \\
$\mu^0$ & $\mu^0$ & $\mu^0$
\end{tabular}
}

\end{tabular}

\caption{Preferences and matchings of \Cref{ex3}.}
\label{exB}
\end{figure}

First notice that $\mu^0_{-a_1}$ is  \(\mathbf C^{0}\)-dominated. Given any possible conjectures that $b_1$ and $a_2$ may hold when they are single, they prefer matching together to remaining single. By (1) in \Cref{dominance} agent $a_1$  prefers remaining alone to matching with $b_1$  for all conjectures in $\mathbf{C}^1(a_1,\succsim)$ and all conjectures in $\mathbf{C}^1(a_1,b_1,\succsim)$.  Hence $\mu^1$  is $\mathbf{C}^1$-dominated. Therefore, $\mu^1$ is not a rationalizable matching. 

Nevertheless, $\mu^1$ is P-stable. Notice that $\mu^1$ is stable according to \citet{Sasaki1996} too by \Cref{nestedness} since $\mathbf C^P_\mu\subseteq\mathbf C^0$. Indeed, in $\mu^1$ agent $b_1$ gets her best match. If $a_1$ deviates by remaining alone he believes that $\mu^0$ would occur and therefore $a_1$'s deviation is not profitable.  
\end{example}

\section{Epistemic Foundation}
\label{foundation}

The epistemic approach to non-cooperative games extends the analysis of
decision-making under uncertainty to strategic settings \citep{Tan}.
Under uncertainty, a decision problem specifies the possible outcomes,
the relevant states of the environment, the agent's preferences, and her
beliefs about those states. A choice rule can then be used to predict or
recommend behavior.

In a strategic setting, the usual ingredients of a game---players,
actions, outcomes, and preferences---are therefore not sufficient.
One must also specify what players believe about the actions and beliefs
of others. Inspired by  this approach, we begin by describing agents'
beliefs about possible matchings and about the beliefs of other agents.

\paragraph*{Belief Hierarchies.}
The essential element of epistemic analysis is the notion of belief hierarchies which specify the players' beliefs over a basic
space of uncertainty, their beliefs over the beliefs of the others, and
so on. In what follows we elaborate a notion of belief hierarchies which suits matching games with externalities and that  is used to
define  the notion of common belief in pairwise rationality.
Beliefs are partner-dependent:  the basic space of uncertainty for each agent $i$ depends on $i$'s potential partner. This departs  from the canonical epistemic  model \citep[e.g.][]{Dekel}.
Then, we  denote $X_{i}= \Pi_{j \in \mathcal{P}_i} \mathcal{M}_{-i,j}$ the {\bf basic space of uncertainty} for any agent $i\in N$. 
Next, for any $i\in N$, we define the set of   $n$-order beliefs of $i$ recursively.
Let $X^1_i=X_i$ be the set of first order beliefs for agent $i$. Suppose $X^{n-1}_i$ has been defined for all $i\in N$. Then $X_i^{n}=\Pi_{j \in \mathcal{P}_i} \Pi_{k \in N \setminus \{i,j\}}X^{n-1}_k$.  A {\bf belief hierarchy} for $i$ is an element of $\mathcal{H}_i= \Pi_{n=1}^\infty X^{n}_i$.
We find it convenient to define the set $\mathcal{H}^{>1}_i= \Pi_{n=1}^\infty X^{n+1}_i$ which consists of   hierarchies  of  higher order beliefs. Note that the collection of belief hierarchies can be rewritten as $\mathcal{H}_i= X_i^1 \times \mathcal{H}_i^{>1}$.

\begin{remark}[Partner-dependence of belief hierarchies]
Unlike the canonical epistemic model of non-cooperative games,
the space of uncertainty $X_i$ is indexed by the potential partner
$j\in\mathcal{P}_i$: as $j$ varies, the residual agents---and hence
the relevant uncertainty---change.
There is therefore no single ``belief of agent~$i$'': rather,
$i$ holds a family of beliefs, one per potential partner.
This partner-dependence is a direct consequence of the cooperative
nature of matching and propagates to every order of the belief
hierarchy, making the construction above irreducible to the
canonical case of \citet{Dekel}.
\end{remark}

\paragraph*{Types.}
As \citet{Harsanyi} aptly noted, type structures provide an alternative way to generate hierarchies of beliefs.
Let $T_i$ be  the set of epistemic types of $i$ with generic element $t_i$, for every $i,j \in N$ such that $i \neq j$, $T_i \cap T_j= \emptyset$. We write $T_{-i}=\Pi_{j\in (  N)\setminus \{i\}}T_j$ with generic element $t_{-i}$, and, for any nonempty $K\subseteq N$, $T_{K}=\Pi_{i\in K}T_i$, and $T_{-K}=\Pi_{j\in (  N)\setminus K}T_j$. If $i,j \in N$ we write $T_{-i,j}=\Pi_{k\in (  N)\setminus \{i,j\}}T_k$ with generic element $t_{-i,j}$.
We also write $T=\Pi_{i\in   N}T_i$ with generic  element $t$.

For all $i \in   N$, the  {\bf belief  map} $ \theta_i:T_i\longrightarrow \Pi_{j\in \mathcal{P}_i}(\mathcal{M}_{-i,j}\times T_{-i,j})$ assigns to each  agent's type a belief hierarchy.

\begin{definition}[Epistemic Type Structure]
    Given a matching game $\Gamma:=\langle A,B,\succsim\rangle$, an  epistemic type structure for 
    $\Gamma$ is a tuple 
    $$\mathscr{T}:=\langle A,B,(T_i,\theta_i)_{i\in   N}\rangle$$  
\end{definition}

Given a type structure $\mathscr{T}$ we denote $\varphi_i:T_i\longrightarrow  \mathcal{H}_i$ the function mapping types into belief hierarchies.
The type structure $\mathscr{T}$ is {\bf complete} if the maps $\varphi_i$ are onto.

Because a type’s first-order and higher order  beliefs play a particularly important role, it is
convenient to introduce specific notation for them.

\begin{definition}[Beliefs]
  Let $\mathscr{T}:=\langle A,B,(T_i,\theta_i)_{i\in   N}\rangle$ be an epistemic  type structure, the \textbf{first order belief}  for agent $i$ of type $t_i$ is $\phi_i(t_i)\equiv \Pi_ {j \in \mathcal{P}_i} \ {\rm proj}_{\mathcal{M}_{-i,j}}\theta_i(t_i)$. The \textbf{higher order beliefs}  for agent $i$ of type $t_i$ are $\psi_i(t_i)\equiv \Pi_ {j \in \mathcal{P}_i} \ {\rm proj}_{T_{-i,j}}\theta_i(t_i)$ (where {\rm proj} denotes the projection operator as
canonically defined).
\end{definition}

\noindent {\bf Notation.}
 Abusing notation we omit the subscript $i$ from the first and higher order beliefs, and  for any $i,j,k,h\in N$ we write
 \begin{itemize}
     \item $\phi(t_i)^j\in \mathcal{M}_{-i,j}$ the first order belief of $i$ having type $t_i$ when she is matched with $j$.
\item $\psi(t_i)_k^j\in T_k$ the type of $k$ conjectured by $i$ of type $t_i$ when $i$ is matched with $j$.
\item $\phi(\psi(t_i)_k^j)^h\in \mathcal{M}_{-k,h}$ the first order belief of $k$ having type $\psi(t_i)_k^j$ when she is matched with $h$.
     \end{itemize}
     For clarity, note that partners are indicated with superscripts, while conjectured types appear as subscripts.

\subsection*{Characterizations}

An {\bf epistemic condition} is any subset of $\mathcal{M}\bigtimes T$. 
An \textbf{event $E_{-i}$ for an agent $i \in N$} is a subset of  $\Omega_{-i}\equiv\Pi_{j \in \mathcal{P}_i}(\mathcal{M}_{-i,j}\times T_{-i,j})$. The collection of all events for an agent $i\in N$ is thus $\powerset(\Omega_{-i})$. We call an \textbf{epistemic event}\footnote{Unlike the canonical framework, our characterization  involves epistemic conditions and events that do not coincide. This separation permits a more nuanced foundation of the solution concepts. This is because agent 
$i$’s basic space of uncertainty depends on her potential partner, whereas the solution concepts are defined solely over matchings.}  an element of $\Pi_{i  \in N}\Omega_{-i}$. 

The next definition states that a matching is pairwise rational for a coalition 
$S$ of agents and a profile of types for each agent in the coalition if, given the agents' first-order beliefs, no matched agent in 
$S$ prefers to remain alone, and no pair of agents in 
$S$ who are not matched to each other would prefer to match together.

\begin{definition}[Pairwise Rationality]\label{PR}
     Let $\left\langle A,B,(T_i,\theta_i)_{i\in   N}\right\rangle$ be an epistemic  type structure for  $\langle A,B,\succsim \rangle$ and fix a  $\mu \in \mathcal{M}$, $S \subseteq  N$ and $t_S \in T_S$. We say that $\mu$ is \textbf{pairwise rational} for $(S,t_S)$ if 
     \begin{enumerate}
         \item Individual Rationality: $\nexists k \in S$  such that   $\phi(t_k)^k+(k) \succ_k \phi(t_k)^{\mu(k)}+(k,\mu(k))$ 
         \item Unblocking: $\nexists (a,b) \in (A \cap S) \bigtimes (B \cap S)$ such that  $\phi(t_a)^{b}+(a,b)\succ_a \phi(t_a)^{\mu(a)}+(a,\mu(a)) $ and $ \phi(t_b)^{a}+(b,a)\succ_b \phi(t_b)^{\mu(b)}+(b,\mu(b))$.
     \end{enumerate}
\end{definition}

\begin{remark}
Each epistemic type holds a single (first-order) belief
for each potential partner; no probabilities are involved. This is without
loss of generality. Suppose instead that, for every partner
$j\in\mathcal P_i$, agent $i$ holds a probability measure
$p_i^j\in\Delta(\mathcal M_{-i,j})$ and evaluates each alternative (which
partner to match with) by expected utility with respect to some utility
function $u_i$ representing $\succsim_i$. Since beliefs are indexed by the
contemplated partner, the belief attached to one alternative places no
restriction on the belief attached to another. 
Consider  an agent $k\in N$, a matching $\mu$ such that $k \neq \mu(k)$    and fix nonempty sets $X\subseteq\mathcal M_{-k}$ and
$Y\subseteq\mathcal M_{-k,\mu(k)}$. We show that if condition (1) of \Cref{PR} does not hold for agent $k$ with probabilistic beliefs in $\Delta(X)\times\Delta(Y)$ then it does not hold with deterministic beliefs in $X \times Y$. 
 Suppose that remaining single yields agent $k$ a strictly higher expected utility than matching with $\mu(k)$  for every pair of beliefs $(p_k^k,p_k^{\mu(k)})\in\Delta(X)\times\Delta(Y)$,  that is
\[
\min_{p\in\Delta(X)}\sum_{\nu\in X}p(\nu)\,u_k\bigl(\nu+(k)\bigr)
\;>\;
\max_{p'\in\Delta(Y)}\sum_{\nu'\in Y}p'(\nu')\,u_k\bigl(\nu'+(k,\mu(k))\bigr).
\]
Notice that the minimum and the maximum are attained by point masses on an element of
$\operatorname*{arg\,min}_{\nu\in X}u_k(\nu+(k))$ and on an element of
$\operatorname*{arg\,max}_{\nu'\in Y}u_k(\nu'+(k,\mu(k)))$, respectively.
Hence the inequality holds if and only if
$\nu+(k)\succ_k\nu'+(k,\mu(k))$ for all $\nu\in X$ and $\nu'\in Y$. It follows that $\mu$ is not individually rational also with deterministic beliefs in $X \times Y$. The same
argument applies to condition (2), separately for $a$ and for $b$. This reasoning translates also to the properties of common belief in pairwise rationality and belief correctness the we introduce in \Cref{PRCBPR,BC}.
Consequently, our main results of this section: \Cref{epistemic1,epistemic2} continue to hold when belief is
read as probability-one belief and
$\theta_i(t_i) \in E_{-i}$ for some event $E_{-i}$ is read as ``the support of $\theta_i(t_i)$ is contained in
$E_{-i}$'' .
\end{remark}

\begin{remark}
Conjecture-rationalizable stability evaluates a deviation against the
set $\mathbf C^\infty(i,j,\succsim)$, whereas each epistemic type specifies
a single conjecture $\phi(t_i)^j$. Under conservative blocking, the two
formulations are equivalent. A deviation is profitable for agent $i$ only
if it is profitable under every conjecture in
$\mathbf C^\infty(i,j,\succsim)$. Since this set is finite and
$\succsim_i$ is complete, it is enough to consider a least favorable
conjecture for $i$. A type whose first-order belief selects such a
conjecture therefore yields the same blocking test.
\end{remark}

For any pair formed by a matching and a type profile $(\mu,t) \in \mathcal{M} \times T$ we say that $(\mu,t) \in \bsf{PR}$ if $\mu$ is pairwise rational for $(N,t)$. Thus $\bsf{PR}$ is the epistemic condition  capturing  pairwise rationality.
Next, we define the event for player $i$ that consists of all conjectures in which the other players behave rationally.
     For  any $i\in N$ let    
\begin{align*}
 \hspace{-1cm}   {\rm PR}_{-i}\equiv\{(\mu_{-i,j}, t_{-i,j})_{j \in {\mathcal P}_i} \in \Omega_{-i}| \forall j \in \mathcal{P}_i, \mu_{-i,j}+(i,j) \text{ is pairwise  rational for } ( N \setminus \{i,j\}, t_{-i,j})\}.
\end{align*}
the event for $i$ such that  agents $N\setminus\{i,j\}$ are rational.
The epistemic event  ${\rm PR}$ is thus $\Pi_{i \in N} {\rm PR}_{-i}$.
The interactive reasoning is captured by means
of opportune operators. The belief operator $\mathbb{B}_i$ of player $i$ is
defined as
\begin{align*}
    \mathbb{B}_i({\rm PR})=\{t_i\in T_i|\theta_i(t_i)\in {\rm PR}_{-i}\}
\end{align*}
Let   $\mathbb{B}^1_i({\rm PR})=\mathbb{B}_i({\rm PR})$ and define  for  $n>1$
\begin{align*}
    \mathbb{B}^n_i({\rm PR})=\{t_i\in T_i|\psi(t_i)_{k}^j\in \mathbb{B}^{n-1}_k({\rm PR})\  \forall j\in \mathcal{P}_i, \ \forall k\in N\setminus\{i,j\}\} 
\end{align*}

Given this, let
\begin{align*}
    \mathbb{CB}_i({\rm PR})=\bigcap_{n=1}^\infty  \mathbb{B}^n_i({\rm PR}) \ \ \text{   and  }\ \ \mathbb{CB}({\rm PR})=\Pi_{i\in N}\mathbb{CB}_i({\rm PR})
\end{align*}
The following definition states the epistemic condition of {\bf Pairwise Rationality and Common Belief in Pairwise Rationality} ($\bsf{PRCBPR}$).
\begin{definition}[\bsf{PRCBPR}]\label{PRCBPR}
   Let $\left\langle A,B,(T_i,\theta_i)_{i\in   N}\right\rangle$ be an epistemic type structure for $\langle A,B,\succsim \rangle$. The epistemic condition {\bf Pairwise Rationality and Common Belief in Pairwise Rationality} is 
   \begin{align*}
       \bsf{PRCBPR}\subseteq \mathcal{M}\bigtimes T: (\mu,t)\in \bsf{PRCBPR} \Longleftrightarrow (\mu,t) \in \bsf{ PR}\bigwedge t\in \mathbb{CB}({\rm PR})
   \end{align*}
\end{definition}

The following \Cref{epistemic1} and \Cref{epistemic2} provide an epistemic characterization of our solution concepts.

\Cref{epistemic1} characterizes the set of rationalizable matchings as a  behavioral implication of {\bf Pairwise Rationality and Common Belief in Pairwise Rationality}.

\begin{theorem}[Foundation of Rationalizability] Fix  $\langle A,B, \succsim\rangle$.
\label{epistemic1}
\begin{enumerate}
    \item In any epistemic type structure  $\mathscr{T}$  for $\langle A,B,\succsim \rangle$,
    \\ ${\rm proj}_\mathcal{M}\bsf{PRCBPR}\subseteq \mathbf{M}^\infty$
    \item In any complete  epistemic type structure  $\mathscr{T}$  for $\langle A,B,\succsim \rangle$, \\ 
    ${\rm proj}_\mathcal{M}\bsf{PRCBPR}=\mathbf{M}^\infty$
\end{enumerate}
    
\end{theorem}

The epistemic condition of {\bf Pairwise Rationality and Common Belief in Pairwise Rationality} alone is not sufficient to characterize the set of conjecture-rationalizable stable matchings. It must be complemented by the following condition of {\bf Belief Correctness}. The epistemic condition {\bf Belief Correctness}  requires that when $i$ is matched  as in $\mu$ then the first order belief she attaches to $\mu(i)$ is exactly $\mu_{-i,\mu(i)}$.

\begin{definition}[\bsf{BC}]\label{BC}
   Let $\left\langle A,B,(T_i,\theta_i)_{i\in   N}\right\rangle$ be an epistemic type structure for $\langle A,B,\succsim \rangle$.  The epistemic condition {\bf Belief Correctness} is
   \begin{align*} 
\bsf{BC}:=
\left\{
(\mu,t)\in\mathcal M\times T:
\phi(t_i)^{\mu(i)}=\mu_{-i,\mu(i)}
\text{ for every }i\in N
\right\}.
   \end{align*}
\end{definition}

\begin{theorem}[Foundation of Stability with Rationalizable Conjectures]
    \label{epistemic2}
     Fix  $\langle A,B, \succsim\rangle$.
\begin{enumerate}
    \item In any epistemic type structure  $\left\langle A,B,(T_i,\theta_i)_{i\in   N}\right\rangle$  for $\langle A,B,\succsim \rangle$,
    \\ ${\rm proj}_{\mathcal M}\bigl(\bsf{PRCBPR}\cap \bsf{BC}\bigr)
\subseteq \mathbf S^\infty$
    \item In any complete  epistemic type structure  $\mathscr{T}$   for $\langle A,B,\succsim \rangle$,\\
    ${\rm proj}_{\mathcal M}\bigl(\bsf{PRCBPR}\cap \bsf{BC}\bigr)
=
\mathbf S^\infty$  
\end{enumerate}
\end{theorem}

Finally, we show that the notion of P-stability needs an ad hoc belief restriction. 
  The epistemic condition {\bf quasifixed belief} requires that when an agent $i$ deviates, she believes that former partners remain single and other agents perform the same matching.

\begin{definition}[\bsf{QFB}]
  Let $\left\langle A,B,(T_i,\theta_i)_{i\in N}\right\rangle$ be an epistemic
  type structure for $\langle A,B,\succsim\rangle$.  For all $\mu\in\mathcal{M}$
  and all $i,j\in N$ with $j\in\mathcal{P}_i\setminus\{\mu(i)\}$, let
  $\hat{\mu}_{-i,j}\in\mathcal{M}_{-i,j}$ be the matching of
  $N\setminus\{i,j\}$ defined by: $\hat{\mu}_{-i,j}(k)=k$ if
  $k\in\{\mu(i),\mu(j)\}$, and $\hat{\mu}_{-i,j}(k)=\mu(k)$ otherwise.
  The epistemic condition \textbf{Quasifixed Belief} is
  \[
   \bsf{QFB}:=
\left\{
(\mu,t)\in\mathcal M\times T:
\phi(t_i)^j=\hat{\mu}_{-i,j}
\text{ for every } i\in N
\text{ and every } j\in\mathcal P_i\setminus\{\mu(i)\}
\right\}.
  \]
\end{definition}

\begin{theorem}[Foundation of P-stability]
    \label{epistemic3}
     Fix  $\langle A,B, \succsim\rangle$, and let $\mathbf{P}$ be the set of P-stable matchings of $\langle A,B, \succsim\rangle$.
\begin{enumerate}
    \item In any epistemic type structure  $\left\langle A,B,(T_i,\theta_i)_{i\in   N}\right\rangle$  for $\langle A,B,\succsim \rangle$,
    \\ ${\rm proj}_{\mathcal M}
\bigl(\bsf{PR}\cap\bsf{BC}\cap\bsf{QFB}\bigr)
\subseteq \mathbf P$
    \item In any complete  epistemic type structure  $\mathscr{T}$   for $\langle A,B,\succsim \rangle$,\\
    ${\rm proj}_{\mathcal M}
\bigl(\bsf{PR}\cap\bsf{BC}\cap\bsf{QFB}\bigr)= \mathbf P$ 
\end{enumerate}
\end{theorem}
\begin{remark}[Epistemic simplicity of P-stability]
Unlike \Cref{epistemic1} and \Cref{epistemic2},
the foundation of P-stability does not require common belief in
pairwise rationality.
This is not an oversight: \bsf{QFB} already pins down
off-path beliefs rigidly and uniformly, leaving no room for
higher-order reasoning about others' rationality to play any role.
Epistemic simplicity here comes at the cost of behavioral
implausibility, since \bsf{QFB} imposes a very specific and
exogenous counterfactual on every agent.
\end{remark}

The three characterizations reveal a clear hierarchy. 
Rationalizability requires \bsf{PRCBPR}; conjecture-rationalizable stability 
adds \bsf{BC}.  P-stability moves in a different direction: 
it weakens interactive reasoning (dropping common belief) 
but rigidly fixes counterfactual beliefs via \bsf{QFB}.

\section{Discussion}

This paper studies stability in   one-to-one matching markets with externalities. In this setting, a deviation does not
fully determine its consequences. The deviating agents must  form 
conjectures about how the residual market arranges. We discipline these
conjectures through a rationalizability-type of procedure and propose a stability notion based on rationalizable conjectures, alongside a non-equilibrium concept of rationalizable matchings which always exist and admit epistemic characterizations. Two extensions appear
especially important
\medskip

\noindent {\bf Dynamic processes and convergence.}
A longstanding question in matching theory is how stable outcomes emerge in
decentralized markets. 
In the classical one-to-one model without externalities,
\citet{Vandevate} show that randomly selected blocking pairs lead to a stable
matching with probability one. Subsequent work identifies important limits to
this result. \citet{SchipperZhang2025} show that prioritizing mutually optimal
blocking pairs may generate cycles, while \citet{Lauermann2014} show that a
strategic search process need not converge to stable outcomes when the
underlying market has multiple stable matchings.  Most closely related to our
problem, \citet{Chen2020} study a joint process of learning and rematching
under one-sided incomplete information. By including both the allocation and
firms' information in the state of the process, they establish almost-sure
convergence to a stable state.

Our problem is similar to that of \citet{Chen2020} because both matching and
beliefs must evolve. The source of uncertainty, however, is different. In
their model, agents learn about an objectively given profile of workers'
types. In our model, preferences are common knowledge and conjectures concern
the residual matching that would follow a deviation. There is no realized
residual matching to observe before the deviation occurs. Moreover, the
formation of a blocking pair does not determine how the remaining agents
rematch. Imposing the classical rule that former partners become single and
all other matches remain unchanged would simply build the counterfactual of
P-stability into the process.

A dynamic extension must therefore specify the joint evolution of matchings
and partner-dependent conjectures without fixing outsiders' reactions
exogenously. The main questions are whether past observations can
lead to the rationalizable system $\mathbf C^\infty$, and  whether
conjecture-rationalizable stable matchings are absorbing outcomes.
\medskip

\noindent {\bf Coalitional deviations and many-to-one matching.}
Our analysis restricts deviations to individuals and pairs. This
restriction is substantive, but analytically useful: it keeps the notation
parsimonious and isolates the role of conjectures after deviations.
In the classical
one-to-one model without externalities, pairwise stability and the core
coincide \citep{Roth1990}. This equivalence need not survive externalities.
\citet{Sasaki1996} already show that, under conservative conjectures, the conjectural core is contained in the pairwise stable set and that the
inclusion may be strict: the core may be empty even though the pairwise stable
set is nonempty. 
Extending our approach to many-to-one matching involves substantive modeling choices. We discuss two possible approaches below. The first is to extend pairwise blocking to blocking by a firm and its prospective workforce. Under this approach, a deviation is undertaken by a coalition  consisting of a firm $j$ and the set $N_j$ of workers it employs under the deviation. As in our setting, the conjecturing agents are excluded from the residual matchings, which specify the workers employed by the other firms. This approach naturally extends Definition \ref{conjectures}, but it has two limitations.
First, the resulting stability concept would differ from both the core and pairwise stability. It would differ from the core because only coalitions consisting of a single firm and its entire prospective workforce could deviate, thus imposing an exogenous restriction on the set of admissible deviations. It would also differ from pairwise stability because deviations involving a firm and multiple workers would be allowed. Second, allowing such coalitional deviations may jeopardize existence: the set of stable allocations may be empty unless further restrictions are imposed on the preference domain. 

An alternative approach is to retain pairwise stability in the many-to-one setting, allowing deviations only by a firm $j$ and a worker $i$. This would, however, require modifying the definition of a system of conjectures. A conjecture would also need to specify which workers other than $i$ firm $j$ would employ following the deviation. Thus, the deviating pair would need to form conjectures not only about the residual matching but also about the composition of the deviating firm's own workforce. Both approaches are possible, as are other alternatives, but each involves different choices about admissible deviations and the scope of conjectures.  An extension to many-to-one matching therefore requires a separate analysis.

\section*{\hypertarget{Appendix}{Appendix}}

\addcontentsline{toc}{section}{Appendix: Proofs}

For mathematical ease, we introduce the notion of feasible coalition. A coalition $(i,j) \in N \times N$ is feasible  if there exists a matching $\mu\in \mathcal{M}$ such that $\mu(i)=j$.  $\mathcal{S}$ is the set of feasible coalitions.  
\medskip

\noindent{\bf Proof of \Cref{nestedness}}.
We prove the statement by contraposition. Suppose that \(\mu\) is not
\(\mathbf C'\)-stable. We show that \(\mu\) is not \(\mathbf C\)-stable.

There are two cases.

\noindent{\bf [1]} There exists \(k\in N\) such that
$
\nu_{-k}+(k)\succ_k\mu
\quad\text{for every }\nu_{-k}\in\mathbf C'(k,\succsim).
$
Since \(\mathbf C(k,\succsim)\subseteq\mathbf C'(k,\succsim)\), the same
inequality holds for every \(\nu_{-k}\in\mathbf C(k,\succsim)\). Hence \(\mu\)
violates \(\mathbf C\)-individual rationality.

\noindent{\bf [2]} There exists \((a,b)\in A\times B\) 
such that
$
\nu_{-a,b}+(a,b)\succ_a\mu
\quad\text{and}\quad
\nu_{-a,b}+(a,b)\succ_b\mu
$
for every \(\nu_{-a,b}\in\mathbf C'(a,b,\succsim)\).
Since \(\mathbf C(a,b,\succsim)\subseteq\mathbf C'(a,b,\succsim)\), the same
inequalities hold for every \(\nu_{-a,b}\in\mathbf C(a,b,\succsim)\). Hence
\(\mu\) is \(\mathbf C\)-blocked.

Thus, if \(\mu\) is not \(\mathbf C'\)-stable, then \(\mu\) is not
\(\mathbf C\)-stable. Equivalently, every \(\mathbf C\)-stable matching is
\(\mathbf C'\)-stable.
\hfill \(\blacksquare\)

\medskip

\noindent{\bf Proof of \Cref{th1}}. The proof builds on the following lemmata.

\begin{lemma} \label{lemmamonotonicity}
     If for all feasible coalitions $(i,j)$ it holds that $\mathbf{C}(i,j,\succsim)\supseteq \mathbf{C}'(i,j,\succsim)$ then every conjecture that is $\mathbf{C}'$-undominated  is $\mathbf{C}$-undominated.
\end{lemma}

\begin{proof}
We prove the statement by contraposition. Fix a feasible pair \((i,j)\) and a
conjecture \(\mu_{-i,j}\in\mathcal M_{-i,j}\). Suppose that \(\mu_{-i,j}\) is not
\(\mathbf C\)-undominated. We show that it is not \(\mathbf C'\)-undominated.
There are two cases.

\noindent{\bf [1]} There exists \(k\in N\setminus\{i,j\}\) such that
$
\nu_{-k}+(k)\succ_k
\nu'_{-k,\mu_{-i,j}(k)}+\bigl(k,\mu_{-i,j}(k)\bigr)
$
for every \(\nu_{-k}\in\mathbf C(k,\succsim)\) and every
$
\nu'_{-k,\mu_{-i,j}(k)}
\in\mathbf C\bigl(k,\mu_{-i,j}(k),\succsim\bigr).
$
Since \(\mathbf C'\subseteq\mathbf C\) pointwise, the same inequalities hold on
the corresponding \(\mathbf C'\)-sets. Hence \(\mu_{-i,j}\) is not
\(\mathbf C'\)-undominated.

\noindent{\bf [2]} There exists
$
(a,b)\in(A\setminus\{i,j\})\times(B\setminus\{i,j\})
$
such that
$
\nu_{-a,b}+(a,b)\succ_a
\nu'_{-a,\mu_{-i,j}(a)}+\bigl(a,\mu_{-i,j}(a)\bigr)
$
and
$
\nu_{-a,b}+(a,b)\succ_b
\nu''_{-b,\mu_{-i,j}(b)}+\bigl(b,\mu_{-i,j}(b)\bigr)
$
for every \(\nu_{-a,b}\in\mathbf C(a,b,\succsim)\), every
$
\nu'_{-a,\mu_{-i,j}(a)}
\in\mathbf C\bigl(a,\mu_{-i,j}(a),\succsim\bigr),
$
and every
$
\nu''_{-b,\mu_{-i,j}(b)}
\in\mathbf C\bigl(b,\mu_{-i,j}(b),\succsim\bigr).
$
Again, since \(\mathbf C'\subseteq\mathbf C\) pointwise, the same inequalities
hold on the corresponding \(\mathbf C'\)-sets. Hence \(\mu_{-i,j}\) is not
\(\mathbf C'\)-undominated.
\end{proof}
\begin{lemma} \label{LemmaSubset}
For each $(i,j)\in \mathcal{S}$  it holds that $\mathbf{C}^n(i,j,\succsim)\supseteq \mathbf{C}^{n+1}(i,j,\succsim)$ for all $n \in \mathbb{ N}$.  
\end{lemma}
\begin{proof}
    Fix any feasible coalition $(i,j)$.   We  proceed by induction over $n$.
    
     \noindent{\bf ($n=0$).}  $\mathbf{C}^0(i,j,\succsim)\supseteq \mathbf{C}^{1}(i,j,\succsim)$ holds by definition of $\mathbf{C}^0(i,j,\succsim)\equiv\mathcal{M}_{-i,j}$.
     
     \noindent{\bf ($n>0$).} Fix any $n>0$ and suppose that   $\mathbf{C}^{n-1}(i,j,\succsim)\supseteq \mathbf{C}^{n}(i,j,\succsim)$.
Then, by  \Cref{lemmamonotonicity} we have that  every conjecture for $(i,j)$ that is  $\mathbf{C}^{n}$-undominated is  $\mathbf{C}^{n-1}$-undominated as well.  Since, by definition, the set of  $\mathbf{C}^{n}$-undominated conjectures for $(i,j)$ is   $\mathbf{C}^{n+1}(i,j,\succsim)$ and  the set of  $\mathbf{C}^{n-1}$-undominated conjectures for $(i,j)$ is $\mathbf{C}^n(i,j,\succsim)$,   it follows  that $\mathbf{C}^n(i,j,\succsim)\supseteq \mathbf{C}^{n+1}(i,j,\succsim)$ for any  $(i,j)\in \mathcal{S}$.  
\end{proof}

\begin{lemma}
\label{lemmaDknonempty}
Fix $n\in\mathbb N$. If, for every feasible coalition $(i,j)\in\mathcal S$, $\mathbf C^n(i,j,\succsim)\neq\emptyset$, then the set of $\mathbf C^n$-stable matchings is nonempty.
\end{lemma}

\begin{proof}
Fix $n\in\mathbb N$ and suppose that $\mathbf C^n(i,j,\succsim)\neq\emptyset$ for every feasible coalition $(i,j)\in\mathcal S$.

For every feasible coalition $(i,j)\in\mathcal S$, choose a conjecture $c_{ij}\in \mathbf C^n(i,j,\succsim)$ such that $c_{ij}+(i,j)$ is the worst for agent $i$ among all complete matchings induced by conjectures in $\mathbf C^n(i,j,\succsim)$; that is, $\nu+(i,j)\succsim_i c_{ij}+(i,j)$ for every $\nu\in \mathbf C^n(i,j,\succsim)$. Such a conjecture exists because $\mathbf C^n(i,j,\succsim)$ is finite and nonempty.

We now construct an auxiliary matching problem without externalities. The sets of agents are $A$ and $B$. For every agent $i\in N$, define a preference ordering $\succsim_i^n$ over $\mathcal P_i$ as follows: for all $j,h\in\mathcal P_i$, $j\succsim_i^n h$ if and only if $c_{ij}+(i,j)\succsim_i c_{ih}+(i,h)$. Since this is a standard two-sided matching problem, it admits a weakly stable matching. Let $\mu^*$ be such a weakly stable matching.

We first show that, for every $i\in N$, $\mu^*_{-i,\mu^*(i)}\in \mathbf C^n(i,\mu^*(i),\succsim)$. It is enough to show that $\mu^*_{-i,\mu^*(i)}$ is $\mathbf C^n$-undominated for every $i\in N$, because then $\mu^*_{-i,\mu^*(i)}\in \mathbf C^{n+1}(i,\mu^*(i),\succsim)\subseteq \mathbf C^n(i,\mu^*(i),\succsim)$, where the inclusion follows from \Cref{LemmaSubset}.

Suppose, toward a contradiction, that for some $i\in N$, the conjecture $\mu^*_{-i,\mu^*(i)}$ is not $\mathbf C^n$-undominated. There are two cases.

\noindent{\bf [1]} There exists $k\in N\setminus\{i,\mu^*(i)\}$ such that $\nu+(k)\succ_k \nu'+(k,\mu^*(k))$ for every $\nu\in\mathbf C^n(k,\succsim)$ and every $\nu'\in \mathbf C^n(k,\mu^*(k),\succsim)$. In particular, $c_{kk}+(k)\succ_k c_{k,\mu^*(k)}+(k,\mu^*(k))$. Hence $k\succ_k^n\mu^*(k)$, contradicting the individual rationality of $\mu^*$ in the auxiliary matching problem.

\noindent{\bf [2]} There exists $(a,b)\in A\times B$, with $a,b\notin\{i,\mu^*(i)\}$, such that $\nu+(a,b)\succ_a \nu'+(a,\mu^*(a))$ for every $\nu\in\mathbf C^n(a,b,\succsim)$ and every $\nu'\in \mathbf C^n(a,\mu^*(a),\succsim)$, and $\eta+(b,a)\succ_b \eta'+(b,\mu^*(b))$ for every $\eta\in\mathbf C^n(b,a,\succsim)$ and every $\eta'\in \mathbf C^n(b,\mu^*(b),\succsim)$. In particular, $c_{ab}+(a,b)\succ_a c_{a,\mu^*(a)}+(a,\mu^*(a))$ and $c_{ba}+(b,a)\succ_b c_{b,\mu^*(b)}+(b,\mu^*(b))$. Hence $b\succ_a^n\mu^*(a)$ and $a\succ_b^n\mu^*(b)$, contradicting the weak stability of $\mu^*$ in the auxiliary matching problem.
Therefore, for every $i\in N$, $\mu^*_{-i,\mu^*(i)}\in \mathbf C^n(i,\mu^*(i),\succsim)$.
We now prove that $\mu^*$ is $\mathbf C^n$-stable. Suppose, toward a contradiction, that $\mu^*$ is not $\mathbf C^n$-stable. There are two cases.

\noindent{\bf [1]} $\mu^*$ violates $\mathbf C^n$-individual rationality. Then there exists $k\in N$ such that $\nu+(k)\succ_k\mu^*$ for every $\nu\in\mathbf C^n(k,\succsim)$. Since $\mu^*_{-k,\mu^*(k)}\in \mathbf C^n(k,\mu^*(k),\succsim)$, and since $c_{k,\mu^*(k)}$ is worst for $k$ in $\mathbf C^n(k,\mu^*(k),\succsim)$, we have $\mu^*=\mu^*_{-k,\mu^*(k)}+(k,\mu^*(k))\succsim_k c_{k,\mu^*(k)}+(k,\mu^*(k))$. Moreover, $c_{kk}+(k)\succ_k\mu^*$. Hence $c_{kk}+(k)\succ_k c_{k,\mu^*(k)}+(k,\mu^*(k))$, so $k\succ_k^n\mu^*(k)$, contradicting individual rationality of $\mu^*$ in the auxiliary matching problem.

\noindent{\bf [2]} $\mu^*$ is $\mathbf C^n$-blocked by a pair $(a,b)\in A\times B$. Then $\nu+(a,b)\succ_a\mu^*$ for every $\nu\in\mathbf C^n(a,b,\succsim)$, and $\eta+(b,a)\succ_b\mu^*$ for every $\eta\in\mathbf C^n(b,a,\succsim)$. Since $\mu^*_{-a,\mu^*(a)}\in \mathbf C^n(a,\mu^*(a),\succsim)$ and $\mu^*_{-b,\mu^*(b)}\in \mathbf C^n(b,\mu^*(b),\succsim)$, and since $c_{a,\mu^*(a)}$ and $c_{b,\mu^*(b)}$ are worst conjectures in their respective sets, we have $\mu^*=\mu^*_{-a,\mu^*(a)}+(a,\mu^*(a))\succsim_a c_{a,\mu^*(a)}+(a,\mu^*(a))$ and $\mu^*=\mu^*_{-b,\mu^*(b)}+(b,\mu^*(b))\succsim_b c_{b,\mu^*(b)}+(b,\mu^*(b))$. Moreover, $c_{ab}+(a,b)\succ_a\mu^*$ and $c_{ba}+(b,a)\succ_b\mu^*$. Therefore, $c_{ab}+(a,b)\succ_a c_{a,\mu^*(a)}+(a,\mu^*(a))$ and $c_{ba}+(b,a)\succ_b c_{b,\mu^*(b)}+(b,\mu^*(b))$. Thus $b\succ_a^n\mu^*(a)$ and $a\succ_b^n\mu^*(b)$, contradicting the weak stability of $\mu^*$ in the auxiliary matching problem.

Hence $\mu^*$ is $\mathbf C^n$-stable.
\end{proof}

\begin{lemma}
\label{Lem_useful}
If $\mathbf C^n(r,s,\succsim)\neq\emptyset$ for every feasible coalition $(r,s)\in\mathcal S$, then $\mathbf C^{n+1}(i,j,\succsim)\neq\emptyset$ for every feasible coalition $(i,j)\in\mathcal S$. Consequently, $\mathbf C^n(i,j,\succsim)\neq\emptyset$ for every feasible coalition $(i,j)\in\mathcal S$ and every $n\in\mathbb N$.
\end{lemma}

\begin{proof}
Fix a feasible coalition $(i,j)\in\mathcal S$ and suppose that $\mathbf C^n(r,s,\succsim)\neq\emptyset$ for every feasible coalition $(r,s)\in\mathcal S$.
We construct an auxiliary matching game on the set of agents $N\setminus\{i,j\}$. For every agent $\ell\in N\setminus\{i,j\}$ and every feasible partner $h\in\mathcal P_\ell\setminus\{i,j\}$, choose a conjecture $c_{\ell h}\in\mathbf C^n(\ell,h,\succsim)$ such that $c_{\ell h}+(\ell,h)$ is worst for $\ell$ among all complete matchings induced by conjectures in $\mathbf C^n(\ell,h,\succsim)$; that is, $\nu+(\ell,h)\succsim_\ell c_{\ell h}+(\ell,h)$ for every $\nu\in\mathbf C^n(\ell,h,\succsim)$. Such a conjecture exists because $\mathbf C^n(\ell,h,\succsim)$ is finite and nonempty.
Define preferences $\succsim^\circ_\ell$ over $\mathcal P_\ell\setminus\{i,j\}$ as follows: for all $h,h'\in\mathcal P_\ell\setminus\{i,j\}$, $h\succsim^\circ_\ell h'$ if and only if $c_{\ell h}+(\ell,h)\succsim_\ell c_{\ell h'}+(\ell,h')$.
The auxiliary game has no externalities, and therefore admits a stable matching. Denote such a matching by $\lambda_{-i,j}\in\mathcal M_{-i,j}$. We show that $\lambda_{-i,j}\in\mathbf C^{n+1}(i,j,\succsim)$, or equivalently, that $\lambda_{-i,j}$ is $\mathbf C^n$-undominated.
Suppose, toward a contradiction, that $\lambda_{-i,j}$ is not $\mathbf C^n$-undominated. There are two cases.

\noindent{\bf [1]} There exists $k\in N\setminus\{i,j\}$ such that $\nu+(k)\succ_k \nu'+(k,\lambda_{-i,j}(k))$ for every $\nu\in\mathbf C^n(k,\succsim)$ and every $\nu'\in\mathbf C^n(k,\lambda_{-i,j}(k),\succsim)$. In particular, $c_{kk}+(k)\succ_k c_{k,\lambda_{-i,j}(k)}+(k,\lambda_{-i,j}(k))$. Hence $k\succ^\circ_k \lambda_{-i,j}(k)$, contradicting the individual rationality of $\lambda_{-i,j}$ in the auxiliary game.

\noindent{\bf [2]} There exists $(a,b)\in(A\setminus\{i,j\})\times(B\setminus\{i,j\})$ such that $\nu+(a,b)\succ_a \nu'+(a,\lambda_{-i,j}(a))$ for every $\nu\in\mathbf C^n(a,b,\succsim)$ and every $\nu'\in\mathbf C^n(a,\lambda_{-i,j}(a),\succsim)$, and $\eta+(b,a)\succ_b \eta'+(b,\lambda_{-i,j}(b))$ for every $\eta\in\mathbf C^n(b,a,\succsim)$ and every $\eta'\in\mathbf C^n(b,\lambda_{-i,j}(b),\succsim)$. In particular, $c_{ab}+(a,b)\succ_a c_{a,\lambda_{-i,j}(a)}+(a,\lambda_{-i,j}(a))$ and $c_{ba}+(b,a)\succ_b c_{b,\lambda_{-i,j}(b)}+(b,\lambda_{-i,j}(b))$. Hence $b\succ^\circ_a \lambda_{-i,j}(a)$ and $a\succ^\circ_b \lambda_{-i,j}(b)$, contradicting the stability of $\lambda_{-i,j}$ in the auxiliary game.
Therefore, $\lambda_{-i,j}$ is $\mathbf C^n$-undominated. By definition of the rationalization procedure, $\lambda_{-i,j}\in\mathbf C^{n+1}(i,j,\succsim)$. Hence $\mathbf C^{n+1}(i,j,\succsim)\neq\emptyset$.
Since $(i,j)$ was arbitrary, $\mathbf C^{n+1}(i,j,\succsim)\neq\emptyset$ for every feasible coalition $(i,j)\in\mathcal S$. The final statement follows by induction from $\mathbf C^0(i,j,\succsim)=\mathcal M_{-i,j}\neq\emptyset$.
\end{proof}

For each feasible coalition $(i,j)\in\mathcal S$, \Cref{LemmaSubset} implies that
the sequence $\mathbf C^0(i,j,\succsim)\supseteq\mathbf C^1(i,j,\succsim)\supseteq\cdots$
is decreasing; since $\mathbf C^0(i,j,\succsim)=\mathcal M_{-i,j}$ is finite, it
stabilizes. As $\mathcal S$ is finite, there exists $q\in\mathbb N$ such that
$
\mathbf C^q(i,j,\succsim)=\mathbf C^{q+1}(i,j,\succsim)=\mathbf C^\infty(i,j,\succsim)
\qquad\text{for every }(i,j)\in\mathcal S.
$
By \Cref{Lem_useful}, $\mathbf C^q(i,j,\succsim)\neq\emptyset$ for every feasible
coalition; hence, by \Cref{lemmaDknonempty}, the set of $\mathbf C^q$-stable
matchings is nonempty. Since $\mathbf C^q=\mathbf C^\infty$ componentwise,
$\mathbf S^\infty\neq\emptyset$.
\hfill $\blacksquare$

\noindent{\bf Proof of \Cref{th2}.}
   Fix any  feasible coalition $(i,j)\in \mathcal{S}$. To prove $(1)$  we show that  if $\mathbf{C}$ is self-undominated, then $\mathbf{C}(i,j,\succsim)\subseteq \mathbf{C}^n(i,j,\succsim)$ for all $n \in \mathbb{ N}$. We proceed by induction over $n$.
     
     \noindent{\bf ($n=0$).} $\mathbf{C}(i,j,\succsim)\subseteq \mathbf{C}^0(i,j,\succsim)\equiv \mathcal{M}_{-i,j}$.
     
     \noindent{\bf ($n>0$).} Fix any $n>0$ and  suppose that $\mathbf{C}(i,j,\succsim) \subseteq \mathbf{C}^{n-1}(i,j,\succsim)$.
      Note that since $\mathbf{C}(i,j,\succsim)$ is self-undominated all conjectures  of the pair $(i,j)$ are $\mathbf{C}$-undominated. Also, by definition, $\mathbf{C}^{n}(i,j,\succsim)$ contains all conjectures of the pair $(i,j)$that are $\mathbf{C}^{n-1}$ undominated. 
      By \Cref{lemmamonotonicity}, we know that    every conjecture  that is $\mathbf{C}$-undominated is $\mathbf{C}^{n-1}$-undominated too. Then, $\mathbf{C}(i,j,\succsim)\subseteq \mathbf{C}^n(i,j,\succsim)$. \\
     Since, by induction,  $\mathbf{C}(i,j,\succsim)\subseteq \mathbf{C}^n(i,j,\succsim)$ for all $n\in \mathbb{ N}$, then $\mathbf{C}(i,j,\succsim)\subseteq \mathbf{C}^\infty(i,j,\succsim)$.
   Finally, we prove point (2), that is, \(\mathbf C^\infty\) is self-undominated.
Since the set of feasible coalitions is finite and, for each feasible coalition
\((r,s)\in\mathcal S\), the sequence
\(\mathbf C^0(r,s,\succsim)\supseteq \mathbf C^1(r,s,\succsim)\supseteq\cdots\)
is decreasing and finite, there exists \(q\in\mathbb N\) such that, for every
\((r,s)\in\mathcal S\),
\(\mathbf C^q(r,s,\succsim)=\mathbf C^{q+1}(r,s,\succsim)
=\mathbf C^\infty(r,s,\succsim)\).
Fix a feasible coalition \((i,j)\in\mathcal S\). By definition of the procedure,
$
\mathbf C^{q+1}(i,j,\succsim)
=
\{\mu_{-i,j}\in\mathcal M_{-i,j}:
\mu_{-i,j}\text{ is }\mathbf C^q\text{-undominated}\}.
$
Since \(\mathbf C^q=\mathbf C^\infty\) componentwise, this implies
$
\mathbf C^\infty(i,j,\succsim)
=
\{\mu_{-i,j}\in\mathcal M_{-i,j}:
\mu_{-i,j}\text{ is }\mathbf C^\infty\text{-undominated}\}.
$
Thus \(\mathbf C^\infty\) is self-undominated.
\hfill $\blacksquare$
\medskip

\noindent{\bf Proof of \Cref{admiStable}.}
Let \(\mu\in\mathbf S^\infty\). We show that, for every \(n\in\mathbb N\) and
every \(i\in N\),
$
\mu_{-i,\mu(i)}\in\mathbf C^n(i,\mu(i),\succsim).
$
The result follows by taking intersections over \(n\).
We proceed by induction on \(n\).

\noindent $(n=0)$ The claim follows from
$
\mathbf C^0(i,\mu(i),\succsim)=\mathcal M_{-i,\mu(i)}.
$

\noindent ($n>0$)
Suppose that, for every \(r\in N\),
$
\mu_{-r,\mu(r)}\in\mathbf C^n(r,\mu(r),\succsim).
$
Fix \(i\in N\) and let \(j=\mu(i)\). We show that
$
\mu_{-i,j}\in\mathbf C^{n+1}(i,j,\succsim).
$
Equivalently, we show that \(\mu_{-i,j}\) is \(\mathbf C^n\)-undominated.
Suppose, toward a contradiction, that \(\mu_{-i,j}\) is not
\(\mathbf C^n\)-undominated.

\noindent{\bf [1]} There exists \(k\in N\setminus\{i,j\}\)
such that
$
\nu_{-k}+(k)\succ_k
\nu'_{-k,\mu(k)}+(k,\mu(k))
$
for every \(\nu_{-k}\in\mathbf C^n(k,\succsim)\) and every
$
\nu'_{-k,\mu(k)}\in\mathbf C^n(k,\mu(k),\succsim).
$
By the inductive hypothesis,
$
\mu_{-k,\mu(k)}\in\mathbf C^n(k,\mu(k),\succsim).
$
Moreover,
$
\mathbf C^\infty(k,\succsim)\subseteq\mathbf C^n(k,\succsim).
$
Therefore,
$
\nu_{-k}+(k)\succ_k\mu
\quad
\text{for every }\nu_{-k}\in\mathbf C^\infty(k,\succsim),
$
contradicting \(\mathbf C^\infty\)-individual rationality of \(\mu\).

\noindent{\bf [2]} There exists
$
(a,b)\in(A\setminus\{i,j\})\times(B\setminus\{i,j\})
$
 such that
$
\nu_{-a,b}+(a,b)\succ_a
\nu'_{-a,\mu(a)}+(a,\mu(a))
$
and
$
\nu_{-a,b}+(a,b)\succ_b
\nu''_{-b,\mu(b)}+(b,\mu(b))
$
for every \(\nu_{-a,b}\in\mathbf C^n(a,b,\succsim)\), every
\(\nu'_{-a,\mu(a)}\in\mathbf C^n(a,\mu(a),\succsim)\), and every
\(\nu''_{-b,\mu(b)}\in\mathbf C^n(b,\mu(b),\succsim)\).
By the inductive hypothesis,
$
\mu_{-a,\mu(a)}\in\mathbf C^n(a,\mu(a),\succsim)
\quad\text{and}\quad
\mu_{-b,\mu(b)}\in\mathbf C^n(b,\mu(b),\succsim).
$
Since
$
\mathbf C^\infty(a,b,\succsim)\subseteq\mathbf C^n(a,b,\succsim),
$
we obtain
$
\nu_{-a,b}+(a,b)\succ_a\mu
\quad\text{and}\quad
\nu_{-a,b}+(a,b)\succ_b\mu
$
for every \(\nu_{-a,b}\in\mathbf C^\infty(a,b,\succsim)\), contradicting
\(\mathbf C^\infty\)-unblocking of \(\mu\).
Thus \(\mu_{-i,j}\) is \(\mathbf C^n\)-undominated, and hence
$
\mu_{-i,j}\in\mathbf C^{n+1}(i,j,\succsim).
$
This completes the induction.
\hfill \(\blacksquare\)
\medskip

\noindent{\bf Proof of \Cref{Pareto}.}
Fix $\langle A,B,\succsim\rangle$. The proof builds on the following lemma.

\begin{lemma}
\label{Paretoimpr}
For any $\mu\in\mathbf S^\infty$, if $\mu$ is Pareto dominated by $\mu'$, then $\mu'\in\mathbf S^\infty$.
\end{lemma}

\begin{proof}
Let $\mu\in\mathbf S^\infty$ and suppose that there is a matching $\mu'$ such that $\mu'\succsim_i\mu$ for every $i\in N$ and $\mu'\succ_j\mu$ for some $j\in N$. Suppose, toward a contradiction, that $\mu'\notin\mathbf S^\infty$. Then one of the following two cases occurs.

\noindent{\bf [1]} There exists $k\in N$ such that $\nu_{-k}+(k)\succ_k\mu'$ for every $\nu_{-k}\in\mathbf C^\infty(k,\succsim)$. First note that it must be that $\mu(k)\neq k$. Indeed, if $\mu(k)=k$, then by $\mathbf C^\infty$-consistency of $\mu$, $\mu_{-k}\in\mathbf C^\infty(k,\succsim)$. Hence $\mu=\mu_{-k}+(k)\succ_k\mu'$, contradicting $\mu'\succsim_k\mu$. Since $\mu(k)\neq k$, and since $\mu'\succsim_k\mu$, we have $\nu_{-k}+(k)\succ_k\mu'\succsim_k\mu$ for every $\nu_{-k}\in\mathbf C^\infty(k,\succsim)$. Thus $\mu$ violates $\mathbf C^\infty$-individual rationality, contradicting $\mu\in\mathbf S^\infty$.

\noindent{\bf [2]} There exists a pair $(a,b)\in A\times B$ such that $\nu_{-a,b}+(a,b)\succ_a\mu'$ and $\nu_{-a,b}+(a,b)\succ_b\mu'$ for every $\nu_{-a,b}\in\mathbf C^\infty(a,b,\succsim)$. First note that it must be that $\mu(a)\neq b$. Indeed, if $\mu(a)=b$, then by $\mathbf C^\infty$-consistency of $\mu$, $\mu_{-a,b}\in\mathbf C^\infty(a,b,\succsim)$. Hence $\mu=\mu_{-a,b}+(a,b)\succ_a\mu'$ and $\mu=\mu_{-a,b}+(a,b)\succ_b\mu'$, contradicting $\mu'\succsim_a\mu$ and $\mu'\succsim_b\mu$. Since $\mu(a)\neq b$, and since $\mu'\succsim_a\mu$ and $\mu'\succsim_b\mu$, we have $\nu_{-a,b}+(a,b)\succ_a\mu'\succsim_a\mu$ and $\nu_{-a,b}+(a,b)\succ_b\mu'\succsim_b\mu$ for every $\nu_{-a,b}\in\mathbf C^\infty(a,b,\succsim)$. Thus $\mu$ is $\mathbf C^\infty$-blocked by $(a,b)$, contradicting $\mu\in\mathbf S^\infty$.
Hence $\mu'\in\mathbf S^\infty$.
\end{proof}

Since $\mathcal M$ is finite, starting from any conjecture-rationalizable stable matching and iteratively applying \Cref{Paretoimpr}, we reach a conjecture-rationalizable stable matching that is Pareto efficient.
\hfill $\blacksquare$

\medskip


\noindent{\bf Proof of \Cref{thRat}.}
The proof is implied by \Cref{th1} and \Cref{thRef} together.
\hfill $\blacksquare$

\noindent{\bf Proof of \Cref{thRef}.}
Let $\mu\in\mathbf S^\infty$. Suppose, toward a contradiction, that $\mu\notin\mathbf M^\infty$. Then $\mu$ is $\mathbf C^\infty$-dominated. There are two cases.

\noindent{\bf [1]} There exists $k\in N$  such that $\nu_{-k}+(k)\succ_k\nu'_{-k,\mu(k)}+(k,\mu(k))$ for every $\nu_{-k}\in\mathbf C^\infty(k,\succsim)$ and every $\nu'_{-k,\mu(k)}\in\mathbf C^\infty(k,\mu(k),\succsim)$. By \Cref{admiStable}, $\mu_{-k,\mu(k)}\in\mathbf C^\infty(k,\mu(k),\succsim)$. Therefore, $\nu_{-k}+(k)\succ_k\mu$ for every $\nu_{-k}\in\mathbf C^\infty(k,\succsim)$, which violates $\mathbf C^\infty$-individual rationality of $\mu$.

\noindent{\bf [2]} There exists $(a,b)\in A\times B$ such that $\nu_{-a,b}+(a,b)\succ_a\nu'_{-a,\mu(a)}+(a,\mu(a))$ and $\nu_{-a,b}+(a,b)\succ_b\nu''_{-b,\mu(b)}+(b,\mu(b))$ for every $\nu_{-a,b}\in\mathbf C^\infty(a,b,\succsim)$, every $\nu'_{-a,\mu(a)}\in\mathbf C^\infty(a,\mu(a),\succsim)$, and every $\nu''_{-b,\mu(b)}\in\mathbf C^\infty(b,\mu(b),\succsim)$. By \Cref{admiStable}, $\mu_{-a,\mu(a)}\in\mathbf C^\infty(a,\mu(a),\succsim)$ and $\mu_{-b,\mu(b)}\in\mathbf C^\infty(b,\mu(b),\succsim)$. Hence, $\nu_{-a,b}+(a,b)\succ_a\mu$ and $\nu_{-a,b}+(a,b)\succ_b\mu$ for every $\nu_{-a,b}\in\mathbf C^\infty(a,b,\succsim)$, which violates $\mathbf C^\infty$-unblocking of $\mu$.
Thus $\mu\in\mathbf M^\infty$.
\hfill $\blacksquare$
\medskip

The following two lemmas will be used in the proof of \Cref{epistemic1}.

\begin{lemma}
    \label{lemEpi}
    If $t_i\in \mathbb{CB}_i(\rm {PR})$ then $\psi(t_i)_k^j\in \mathbb{CB}_k(\rm{PR})$ for all $j\in \mathcal{P}_i$ and all $k\in N\setminus\{i,j\}$. 
    \end{lemma}
\begin{proof}
    Fix any $i\in N$ and suppose that $t_i\in \mathbb{CB}_i(\rm {PR})$. Then,  it holds that $t_i\in \mathbb{B}^n_i(\rm{PR}) $ for all $n>1$. It follows that by  definition of $\mathbb{B}^n_i(\rm{PR})$ we have that  $\psi(t_i)_k^j\in \mathbb{B}_k^{n-1}(\rm{PR})$ for all $j\in \mathcal{P}_i$  all $k\in N\setminus\{i,j\}$ and all $n\in \mathbb{N}$, that is, $\psi(t_i)_k^j\in \mathbb{CB}_k(\rm{PR})$ for all $j\in \mathcal{P}_i$ and all $k\in N\setminus\{i,j\}$. 
\end{proof}

The following \Cref{lemEpi2} states that if the type of any agent $i$ believes  in pairwise rationality, then the  first order belief of $i$ consists of a sequence of rationalizable conjectures. 
\begin{lemma}
    \label{lemEpi2}
If $t_i\in \mathbb{CB}_i(\rm PR)$ then $\phi(t_i)\in \Pi_{j \in \mathcal P_i}  (\mathbf{C}^\infty(i,j,\succsim))$.
\end{lemma}

\begin{proof}
 We show that  if  $t_i\in \mathbb{CB}_i(\rm PR)$ then $\phi(t_i)\in  \Pi_{j \in \mathcal P_i} (\mathbf{C}^n(i,j,\succsim))$ for all $n\in \mathbb{N}$.
 Fix any $i\in N$ and suppose that $t_i\in \mathbb{CB}_i(\rm {PR})$.
 We proceed by induction over $n\in \mathbb{N}$.
\medskip

    \noindent {\bf ($n=0$).}  $\phi(t_i)\in \Pi_{j\in \mathcal{P}_i} (\mathbf{C}^0(i,j,\succsim))$ holds by definition of $\mathbf{C}^0$ which contains all possible conjectures.
    \medskip
    
    \noindent {\bf ($n>0$).} Fix any $n>0$ and suppose that  for all $i\in N$ if  $t_i\in \mathbb{CB}_i(\rm  PR)$ then $\phi(t_i)\in \Pi_{j\in \mathcal{P}_i} (\mathbf{C}^{n-1}(i,j,\succsim))$. 
Fix any $j \in \mathcal{P}_i$. Since the choice of $j$ was arbitrary, it is sufficient to  show that   $\phi(t_i)^j\in \mathbf{C}^{n}(i,j,\succsim)$, i.e., $\phi(t_i)^j$ is $\mathbf{C}^{n-1}$-undominated. Define $\mu_{-i,j}= \phi(t_i)^j$  and suppose, toward a contradiction, that $\mu_{-i,j}$ is not $\mathbf{C}^{n-1}$-undominated. There are  two (non-mutually exclusive) cases.

 \noindent{\bf[1]} There is a $k\in   N$ with $i\neq k\neq j$  such that $\nu_{-k}+(k)\succ_k \nu'_{-k,\mu_{-i,j}(k)}+(k,\mu_{-i,j}(k)) $
    for every $\nu_{-k}\in \mathbf{C}^{n-1}(k,\succsim)$ and every $\nu'_{-k,\mu_{-i,j}(k)}\in \mathbf{C}^{n-1}(k,\mu_{-i,j}(k),\succsim)$. Since by assumption $t_i\in \mathbb{CB}_i(\rm PR)$, \Cref{lemEpi} applies and we have that  $\psi(t_i)_k^j\in \mathbb{CB}_k(\rm{PR})$. Moreover,  by the inductive hypothesis $\phi(\psi(t_i)_k^j)^h \in  \mathbf{C}^{n-1}(k,h,\succsim)$ for all  ${h \in \mathcal{P}_k} $. In particular, $\phi(\psi(t_i)_k^j)^k \in  \mathbf{C}^{n-1}(k,\succsim)$ and $\phi(\psi(t_i)_k^j)^{\mu_{-i,j}(k)} \in  \mathbf{C}^{n-1}(k,\mu_{-i,j}(k),\succsim)$.
    This, together with the contradiction hypothesis, implies that
   $ \phi(\psi(t_i)^j_k)^k+(k)
\succ_k
\phi(\psi(t_i)^j_k)^{\mu_{-i,j}(k)}+(k,\mu_{-i,j}(k))$.
       However, since $t_i \in \mathbb{CB}_i(\rm PR)$, we have $(\phi(t_i),\psi(t_i))=\theta_i(t_i)\in {\rm PR}_{-i}$. This means that  $\phi(t_i)^{j}$ is pairwise rational for $(N \setminus \{i,j\}, \psi(t_i)^{j})$, so by Point $1$ in \Cref{PR} $\nexists k  \in N \setminus \{i,j\}$ such that $ \phi(\psi(t_i)^j_k)^k+(k)
\succ_k
\phi(\psi(t_i)^j_k)^{\mu_{-i,j}(k)}+(k,\mu_{-i,j}(k))$, a contradiction.
    
    \noindent{\bf [2]} There is a pair $(a,b)\in (A\setminus\{i,j\})\times (B\setminus \{i,j\})$ such that $\nu_{-a,b}+(a,b)\succ_a \nu'_{-a,\mu_{-i,j}(a)}+(a,\mu_{-i,j}(a))$  and $\nu_{-a,b}+(a,b)\succ_b \nu''_{-b,\mu_{-i,j}(b)}+(b,\mu_{-i,j}(b))$ for every $\nu_{-a,b}\in \mathbf{C}^{n-1}(a,b,\succsim)$,   every $\nu'_{-a,\mu_{-i,j}(a)}\in \mathbf{C}^{n-1}(a,\mu_{-i,j}(a),\succsim)$ and  every $\nu''_{-b,\mu_{-i,j}(b)}\in \mathbf{C}^{n-1}(b,\mu_{-i,j}(b),\succsim)$. Since by assumption $t_i\in \mathbb{CB}_i(\rm PR)$, \Cref{lemEpi} applies and we have that $\psi(t_i)_a^j\in \mathbb{CB}_a(\rm PR)$ and $\psi(t_i)_b^j\in \mathbb{CB}_b(\rm PR)$. 
    Moreover,   by the inductive hypothesis,  
    $\phi(\psi(t_i)_a^j)^b \in  \mathbf{C}^{n-1}(a,b,\succsim)$, $\phi(\psi(t_i)_a^j)^{\mu_{-i,j}(a)} \in  \mathbf{C}^{n-1}(a,\mu_{-i,j}(a),\succsim)$, $\phi(\psi(t_i)_b^j)^a \in  \mathbf{C}^{n-1}(b,a,\succsim)$, and $\phi(\psi(t_i)_b^j)^{\mu_{-i,j}(b)} \in  \mathbf{C}^{n-1}(b,\mu_{-i,j}(b),\succsim)$.
    This, together with the contradiction hypothesis, implies that
    $\phi(\psi(t_i)^j_a)^b+(a,b)
\succ_a
\phi(\psi(t_i)^j_a)^{\mu_{-i,j}(a)}+(a,\mu_{-i,j}(a))$ and   $\phi(\psi(t_i)_b^j)^{a}+(b,a) \succ_b \phi(\psi(t_i)_b^j)^{\mu_{-i,j}(b)}+(b,\mu_{-i,j}(b))$. 
      However, since $t_i \in \mathbb{CB}_i(\rm PR)$, we have $(\phi(t_i),\psi(t_i))=\theta_i(t_i)\in {\rm PR}_{-i}$. Then, $\phi(t_i)^j$ is pairwise rational for $(N \setminus \{i,j\}, \psi(t_i)^j)$, so by Point $2$ in \Cref{PR} there is no $(a,b) \in (A \setminus \{i,j\}) \times  (B \setminus \{i,j\})$ such that $\phi(\psi(t_i)^j_a)^b+(a,b)
\succ_a
\phi(\psi(t_i)^j_a)^{\mu_{-i,j}(a)}+(a,\mu_{-i,j}(a))$ and   $\phi(\psi(t_i)_b^j)^{a}+(b,a) \succ_b \phi(\psi(t_i)_b^j)^{\mu_{-i,j}(b)}+(b,\mu_{-i,j}(b))$, a contradiction.
\end{proof}
\medskip

\noindent{\bf Proof of \Cref{epistemic1}.}
Fix $\langle A,B,\succsim\rangle$. We start by showing point (1). Take any epistemic type structure $\langle A,B,(T_i,\theta_i)_{i\in N}\rangle$ appended to $\langle A,B,\succsim\rangle$ and take any $(\mu,t)\in\bsf{PRCBPR}$. We show that $\mu\in\mathbf M^\infty$.
Suppose, toward a contradiction, that $\mu\notin\mathbf M^\infty$. Then $\mu$ is $\mathbf C^\infty$-dominated. There are two cases.

\noindent{\bf [1]} There exists $k\in N$ such that $\nu_{-k}+(k)\succ_k\nu'_{-k,\mu(k)}+(k,\mu(k))$ for every $\nu_{-k}\in\mathbf C^\infty(k,\succsim)$ and every $\nu'_{-k,\mu(k)}\in\mathbf C^\infty(k,\mu(k),\succsim)$. Since $(\mu,t)\in\bsf{PRCBPR}$, we have $t\in\mathbb{CB}({\rm PR})$, and hence $t_k\in\mathbb{CB}_k({\rm PR})$. By \Cref{lemEpi2}, $\phi(t_k)^j\in\mathbf C^\infty(k,j,\succsim)$ for every $j\in\mathcal P_k$. In particular, $\phi(t_k)^k\in\mathbf C^\infty(k,\succsim)$ and $\phi(t_k)^{\mu(k)}\in\mathbf C^\infty(k,\mu(k),\succsim)$. Therefore, $\phi(t_k)^k+(k)\succ_k\phi(t_k)^{\mu(k)}+(k,\mu(k))$, contradicting pairwise rationality of $\mu$ for $(N,t)$.

\noindent{\bf [2]} There exists $(a,b)\in A\times B$  such that $\nu_{-a,b}+(a,b)\succ_a\nu'_{-a,\mu(a)}+(a,\mu(a))$ and $\nu_{-a,b}+(a,b)\succ_b\nu''_{-b,\mu(b)}+(b,\mu(b))$ for every $\nu_{-a,b}\in\mathbf C^\infty(a,b,\succsim)$, every $\nu'_{-a,\mu(a)}\in\mathbf C^\infty(a,\mu(a),\succsim)$, and every $\nu''_{-b,\mu(b)}\in\mathbf C^\infty(b,\mu(b),\succsim)$. Since $(\mu,t)\in\bsf{PRCBPR}$, $t_a\in\mathbb{CB}_a({\rm PR})$ and $t_b\in\mathbb{CB}_b({\rm PR})$. By \Cref{lemEpi2}, $\phi(t_a)^j\in\mathbf C^\infty(a,j,\succsim)$ for every $j\in\mathcal P_a$ and $\phi(t_b)^j\in\mathbf C^\infty(b,j,\succsim)$ for every $j\in\mathcal P_b$. In particular, $\phi(t_a)^b\in\mathbf C^\infty(a,b,\succsim)$, $\phi(t_a)^{\mu(a)}\in\mathbf C^\infty(a,\mu(a),\succsim)$, $\phi(t_b)^a\in\mathbf C^\infty(b,a,\succsim)$, and $\phi(t_b)^{\mu(b)}\in\mathbf C^\infty(b,\mu(b),\succsim)$. Therefore, $\phi(t_a)^b+(a,b)\succ_a\phi(t_a)^{\mu(a)}+(a,\mu(a))$ and $\phi(t_b)^a+(b,a)\succ_b\phi(t_b)^{\mu(b)}+(b,\mu(b))$, contradicting pairwise rationality of $\mu$ for $(N,t)$.
Thus $\mu\in\mathbf M^\infty$.
We now prove point (2). Let $\mu\in\mathbf M^\infty$. For every feasible $(i,j)\in\mathcal S$, choose conjectures $\mu^+_{-i,j},\mu^-_{-i,j}\in\mathbf C^\infty(i,j,\succsim)$ such that $\mu^+_{-i,j}+(i,j)\succsim_i\nu_{-i,j}+(i,j)\succsim_i\mu^-_{-i,j}+(i,j)$ for every $\nu_{-i,j}\in\mathbf C^\infty(i,j,\succsim)$. Thus, $\mu^+_{-i,j}$ is a best rationalizable conjecture for $i$ conditional on being matched with $j$, and $\mu^-_{-i,j}$ is a worst one.
For each $i\in N$, first construct a belief hierarchy
$h_i\in\mathcal H_i$ as follows.
At first order, agent $i$ conjectures
$\mu^+_{-i,\mu(i)}$ when matched with $\mu(i)$, and
$\mu^-_{-i,j}$ when matched with any
$j\in\mathcal P_i\setminus\{\mu(i)\}$.
Construct higher-order beliefs recursively using the same rule.
At each higher order, each conjectured agent evaluates her
assigned partner under a best rationalizable conjecture and
each alternative under a worst rationalizable conjecture,
relative to the matching conjectured at the preceding level.
By completeness of the type structure, each hierarchy $h_i$
is represented by some type $t_i\in T_i$, that is,
$\varphi_i(t_i)=h_i$.
The resulting type profile $t$ satisfies
$
\phi(t_i)^{\mu(i)}=\mu^+_{-i,\mu(i)}$ and 
$\phi(t_i)^j=\mu^-_{-i,j}$
for every $j\in\mathcal P_i\setminus\{\mu(i)\}$.
We first show that $\mu$ is pairwise rational for $(N,t)$. Suppose not. If individual rationality is violated, then there exists $k\in N$ with $\mu(k)\neq k$ such that $\phi(t_k)^k+(k)\succ_k\phi(t_k)^{\mu(k)}+(k,\mu(k))$. By construction, $\phi(t_k)^k=\mu^-_{-k,k}$ and $\phi(t_k)^{\mu(k)}=\mu^+_{-k,\mu(k)}$. Hence $\nu_{-k}+(k)\succ_k\nu'_{-k,\mu(k)}+(k,\mu(k))$ for every $\nu_{-k}\in\mathbf C^\infty(k,\succsim)$ and every $\nu'_{-k,\mu(k)}\in\mathbf C^\infty(k,\mu(k),\succsim)$, contradicting $\mu\in\mathbf M^\infty$. If pairwise rationality is violated, then there exists $(a,b)\in A\times B$ with $\mu(a)\neq b$ such that $\phi(t_a)^b+(a,b)\succ_a\phi(t_a)^{\mu(a)}+(a,\mu(a))$ and $\phi(t_b)^a+(b,a)\succ_b\phi(t_b)^{\mu(b)}+(b,\mu(b))$. By construction, this implies that the pair $(a,b)$ $\mathbf C^\infty$-dominates $\mu$, again contradicting $\mu\in\mathbf M^\infty$. Thus $(\mu,t)\in\bsf{PR}$.
Finally, by the recursive construction of higher-order beliefs and by the self-undominance of $\mathbf C^\infty$, every type conjectured at every order also satisfies pairwise rationality. Hence $t\in\mathbb{CB}({\rm PR})$. Therefore, $(\mu,t)\in\bsf{PRCBPR}$, and so $\mu\in{\rm proj}_{\mathcal M}\bsf{PRCBPR}$.
\hfill $\blacksquare$
\medskip

\noindent{\bf Proof of \Cref{epistemic2}.}
Fix $\langle A,B,\succsim\rangle$. We start by showing point (1). Take any epistemic type structure $\langle A,B,(T_i,\theta_i)_{i\in N}\rangle$ appended to $\langle A,B,\succsim\rangle$ and take any $(\mu,t)\in\bsf{PRCBPR}\cap\bsf{BC}$. We show that $\mu\in\mathbf S^\infty$.
Suppose, toward a contradiction, that $\mu\notin\mathbf S^\infty$. There are two cases.

\noindent{\bf [1]} There exists $k\in N$  such that $\nu_{-k}+(k)\succ_k\mu$ for every $\nu_{-k}\in\mathbf C^\infty(k,\succsim)$. Since $(\mu,t)\in\bsf{BC}$, $\phi(t_k)^{\mu(k)}+(k,\mu(k))=\mu$. Since $(\mu,t)\in\bsf{PRCBPR}$, $t_k\in\mathbb{CB}_k({\rm PR})$, and by \Cref{lemEpi2}, $\phi(t_k)^k\in\mathbf C^\infty(k,\succsim)$. Therefore, $\phi(t_k)^k+(k)\succ_k\phi(t_k)^{\mu(k)}+(k,\mu(k))$, contradicting pairwise rationality.

\noindent{\bf [2]} There exists $(a,b)\in A\times B$ such that $\nu_{-a,b}+(a,b)\succ_a\mu$ and $\nu_{-a,b}+(a,b)\succ_b\mu$ for every $\nu_{-a,b}\in\mathbf C^\infty(a,b,\succsim)$. Since $(\mu,t)\in\bsf{BC}$, $\phi(t_a)^{\mu(a)}+(a,\mu(a))=\mu$ and $\phi(t_b)^{\mu(b)}+(b,\mu(b))=\mu$. Since $(\mu,t)\in\bsf{PRCBPR}$, by \Cref{lemEpi2}, $\phi(t_a)^b\in\mathbf C^\infty(a,b,\succsim)$ and $\phi(t_b)^a\in\mathbf C^\infty(b,a,\succsim)$. Therefore, $\phi(t_a)^b+(a,b)\succ_a\phi(t_a)^{\mu(a)}+(a,\mu(a))$ and $\phi(t_b)^a+(b,a)\succ_b\phi(t_b)^{\mu(b)}+(b,\mu(b))$, contradicting pairwise rationality.
Thus $\mu\in\mathbf S^\infty$.
We now prove point (2). Let $\mu\in\mathbf S^\infty$. For every feasible $(i,j)\in\mathcal S$, choose conjectures $\mu^+_{-i,j},\mu^-_{-i,j}\in\mathbf C^\infty(i,j,\succsim)$ such that $\mu^+_{-i,j}+(i,j)\succsim_i\nu_{-i,j}+(i,j)\succsim_i\mu^-_{-i,j}+(i,j)$ for every $\nu_{-i,j}\in\mathbf C^\infty(i,j,\succsim)$. By \Cref{admiStable}, $\mu_{-i,\mu(i)}\in\mathbf C^\infty(i,\mu(i),\succsim)$ for every $i\in N$.
For each agent $i$, construct a belief hierarchy $h_i$ as follows.
At first order, she correctly conjectures the residual matching
$\mu_{-i,\mu(i)}$ when matched with $\mu(i)$.
For every other partner $j$, she selects a worst rationalizable
conjecture, $\mu^-_{-i,j}$.
These conjectures belong to $\mathbf C^\infty$ by
\Cref{admiStable} and by the choice of $\mu^-_{-i,j}$.
Specify all higher-order beliefs using the recursive
construction in the proof of \Cref{epistemic1}.
By completeness of the type structure, each hierarchy $h_i$
is represented by a type $t_i\in T_i$, with
$\varphi_i(t_i)=h_i$.
The resulting type profile $t$ satisfies
$\phi(t_i)^{\mu(i)}=\mu_{-i,\mu(i)}$ for every $i\in N$.
Hence $(\mu,t)\in\bsf{BC}$.
We show that $(\mu,t)\in\bsf{PR}$.
Suppose, toward a contradiction, that
$(\mu,t)\notin\bsf{PR}$. There are two cases.

\noindent{\bf [1]}
There exists $k\in N$ with $\mu(k)\neq k$ such that
$
\phi(t_k)^k+(k)
\succ_k
\phi(t_k)^{\mu(k)}+(k,\mu(k))
=\mu.
$
By construction, $\phi(t_k)^k$ is a worst rationalizable
conjecture for $k$ conditional on remaining single. Hence
$
\nu_{-k}+(k)
\succsim_k
\phi(t_k)^k+(k)
\succ_k\mu
$
for every $\nu_{-k}\in\mathbf C^\infty(k,\succsim)$.
This contradicts $\mathbf C^\infty$-individual rationality
of $\mu$.

\noindent{\bf [2]}
There exists $(a,b)\in A\times B$ with $\mu(a)\neq b$
such that
$
\phi(t_a)^b+(a,b)\succ_a\mu,
\qquad
\phi(t_b)^a+(b,a)\succ_b\mu.
$
By construction, these conjectures are worst for $a$ and $b$,
respectively, conditional on matching with each other.
Therefore, every
$\nu_{-a,b}\in\mathbf C^\infty(a,b,\succsim)$ satisfies
$
\nu_{-a,b}+(a,b)
\succsim_a \phi(t_a)^b+(a,b)\succ_a\mu$ and 
$\nu_{-a,b}+(a,b)
\succsim_b \phi(t_b)^a+(b,a)\succ_b\mu.
$
Thus $(a,b)$ $\mathbf C^\infty$-blocks $\mu$, contradicting
$\mu\in\mathbf S^\infty$.

Both cases lead to a contradiction.
Therefore, $(\mu,t)\in\bsf{PR}$.
The recursive construction of higher-order beliefs implies $t\in\mathbb{CB}({\rm PR})$. Hence $(\mu,t)\in\bsf{PRCBPR}\cap\bsf{BC}$, and so $\mu\in{\rm proj}_{\mathcal M}(\bsf{PRCBPR}\cap\bsf{BC})$.
\hfill $\blacksquare$

\medskip

\noindent{\bf Proof of \Cref{epistemic3}.}
Fix $\langle A,B,\succsim\rangle$. We start by showing point (1). Take any epistemic type structure $\langle A,B,(T_i,\theta_i)_{i\in N}\rangle$ appended to $\langle A,B,\succsim\rangle$ and take any $(\mu,t)\in\bsf{PR}\cap\bsf{BC}\cap\bsf{QFB}$. We show that $\mu\in\mathbf P$. Suppose, toward a contradiction, that $\mu\notin\mathbf P$. There are two cases.

\noindent{\bf [1]} There exists $k\in N$ with $\mu(k)\neq k$ such that $\mu^P_{-k}+(k)\succ_k\mu$, where $\mu^P_{-k}\in\mathbf C^P_\mu(k,\succsim)$ is the P-stability conjecture following the unilateral deviation of $k$. By belief correctness, $\phi(t_k)^{\mu(k)}+(k,\mu(k))=\mu$. By quasifixed belief, $\phi(t_k)^k=\mu^P_{-k}$. Hence $\phi(t_k)^k+(k)\succ_k\phi(t_k)^{\mu(k)}+(k,\mu(k))$, contradicting pairwise rationality.

\noindent{\bf [2]} There exists $(a,b)\in A\times B$ with $\mu(a)\neq b$ such that $\mu^P_{-a,b}+(a,b)\succ_a\mu$ and $\mu^P_{-a,b}+(a,b)\succ_b\mu$, where $\mu^P_{-a,b}\in\mathbf C^P_\mu(a,b,\succsim)$ is the P-stability conjecture following the deviation of $(a,b)$. By belief correctness, $\phi(t_a)^{\mu(a)}+(a,\mu(a))=\mu$ and $\phi(t_b)^{\mu(b)}+(b,\mu(b))=\mu$. By quasifixed belief, $\phi(t_a)^b=\mu^P_{-a,b}$ and $\phi(t_b)^a=\mu^P_{-a,b}$. Hence $\phi(t_a)^b+(a,b)\succ_a\phi(t_a)^{\mu(a)}+(a,\mu(a))$ and $\phi(t_b)^a+(b,a)\succ_b\phi(t_b)^{\mu(b)}+(b,\mu(b))$, contradicting pairwise rationality.

Thus $\mu\in\mathbf P$.
We now prove point (2). Let $\mu\in\mathbf P$. By completeness, construct a type profile $t\in T$ such that, for every $i\in N$, $\phi(t_i)^{\mu(i)}=\mu_{-i,\mu(i)}$, and for every $j\in\mathcal P_i\setminus\{\mu(i)\}$, $\phi(t_i)^j$ is the unique P-stability conjecture in $\mathbf C^P_\mu(i,j,\succsim)$. By construction, $(\mu,t)\in\bsf{BC}\cap\bsf{QFB}$.
Suppose, toward a contradiction, that $(\mu,t)\notin\bsf{PR}$. There are two cases.

\noindent{\bf [1]} There exists $k\in N$  such that $\phi(t_k)^k+(k)\succ_k\phi(t_k)^{\mu(k)}+(k,\mu(k))$. By construction, $\phi(t_k)^k$ is the unique conjecture in $\mathbf C^P_\mu(k,\succsim)$ and $\phi(t_k)^{\mu(k)}+(k,\mu(k))=\mu$. Hence $\mu$ is not P-stable, a contradiction.

\noindent{\bf [2]} There exists $(a,b)\in A\times B$ such that $\phi(t_a)^b+(a,b)\succ_a\phi(t_a)^{\mu(a)}+(a,\mu(a))$ and $\phi(t_b)^a+(b,a)\succ_b\phi(t_b)^{\mu(b)}+(b,\mu(b))$. By construction, $\phi(t_a)^b$ and $\phi(t_b)^a$ are the relevant P-stability conjectures following the deviation of $(a,b)$, while $\phi(t_a)^{\mu(a)}+(a,\mu(a))=\phi(t_b)^{\mu(b)}+(b,\mu(b))=\mu$. Hence $\mu$ is not P-stable, a contradiction.
\hfill $\blacksquare$

\let\OLDthebibliography\thebibliography
\renewcommand\thebibliography[1]{
  \OLDthebibliography{#1}
  \setlength{\parskip}{0pt}
  \setlength{\itemsep}{0pt plus 0.3ex}
}


\end{document}